\documentclass[12pt]{article}
\usepackage{longtable,supertabular}
\usepackage{amsmath}
\usepackage{amssymb}
\usepackage{latexsym}
\usepackage{a4wide,color}
\usepackage{multirow}
\usepackage{ntheorem}
\usepackage[square,numbers,sort]{natbib}
\usepackage[title]{appendix}
\usepackage{array}
\usepackage{makecell}
\usepackage{bm}
\usepackage{graphicx}
\usepackage{booktabs}
\usepackage{multirow}

\newtheorem{theorem}{Theorem}[section]
\newtheorem{corollary}{Corollary}[section]
\newtheorem{proposition}{Proposition}[section]
\newtheorem{lemma}{Lemma}[section]
\newtheorem{definition}{Definition}[section]
\newtheorem*{proof}{Proof.}

\newtheorem{example}{Example}[section]

\newtheorem{remark}{Remark}[section]

\newcommand{\F}{\mathbb{F}}

\newcommand{\C}{{\mathcal{C}}}

\newcommand{\supp}{\mathrm{supp}}
\title{\bf Generalized Hamming Weights and Subcode Support Weight Distributions of Twisted  Reed-Solomon Codes}
\author{Chao Liu, Hao Chen, Dabin Zheng and Yaqi Chen
	\thanks{Chao Liu and Hao Chen  are with the College of Information Science and
		Technology, Jinan University, Guangzhou, Guangdong Province, 510632, China (e-mail: haochen@jnu.edu.cn, e-mail: chliuu@163.com).   Dabin Zheng is  with the Hubei Laboratory of Applied Mathematics, Faculty of Mathematics, Hubei University, Wuhan 430062, Hubei, China (dzheng@hubu.edu.cn). Yaqi Chen is with the College of Cyber Security, Jinan University, Guangzhou, Guangdong Province, 510632, China. (e-mail: chenyq@stu.jnu.edu.cn).
		Corresponding author: Dabin Zheng (dzheng@hubu.edu.cn). This research was supported by NSFC Grant 62032009, NSFC Grant 62272148, NSFC Grant 62171479.
}}

\date{}

\begin{document}
	
	\maketitle
	\begin{abstract}
In this paper, we determine the generalized Hamming weight hierarchies of two classes of twisted Reed--Solomon (TRS) codes. For TRS codes with one twist, we first determine all possible values of
the generalized Hamming weights. We then give a complete characterization of the weight hierarchy in terms of the minimum distance of the dual code. Explicit conditions are derived for the minimum distances of the code and its dual, which make the whole weight hierarchy explicitly computable. As special cases, the obtained results recover several known characterizations for TRS codes to be MDS, NMDS codes, and
$s$-gap AMDS codes. For TRS codes with two twists, we similarly determine the complete generalized Hamming weight hierarchy. In this case, the higher generalized Hamming weights are determined by the first two generalized Hamming weights of the dual code, for which explicit rank conditions are also obtained. In addition, we further determine the subcode support weight distributions of this class of single-twist TRS codes. More precisely, for each
$1 \le r \le k$, we derive an explicit formula for the number of $r$-dimensional subcodes with a prescribed support size, thereby giving a complete description of the support distribution of subcodes.

\end{abstract}
	
\newpage

\section{Introduction}

\subsection{Generalized Hamming Weights of Linear Codes}

A linear code $\mathcal{C}$ over $\mathbb{F}_q$ with parameters $[n,k,d] $ is a $k$-dimensional subspace of $\mathbb{F}_q^n$
with length $n$ and minimum Hamming distance $d$. For each $r$ with $1\leq r\leq k$, let $[\mathcal{C}, r]$ be the set of all $r$-dimensional $\mathbb{F}_q$-subspaces of $\mathcal{C}$.
For each $H_r\in [\mathcal{C}, r]$, the {\em support} of $H_r$, denoted by supp($H_r$), is the set of nonzero coordinates of nonzero vectors in $H_r$, i.e.,
\[
{\rm supp}(H_r)=\left\{ i  : c_i\neq 0\,\, {\rm for \,\, some} \,\,(c_1, c_2,\cdots,c_n)\in H_r \right\}.
\]
The $r$-th {\em generalized Hamming weight} (GHW) of $\C$ is defined by
	$$ d_r(\C)= {\rm min} \left\{ |{\rm supp}(H_r)|\,:\, H_r \in [ \C, r] \right\},$$
where $|{\rm supp}(H_r)|$ denotes the cardinality of the set ${\rm supp}(H_r)$. The set $\left\{ d_1(\C), d_2(\C),\cdots, d_k(\C)\right\}$ is called the {\em weight hierarchy}
of $\C$. Note that $d_1(\C)$ is just the minimum Hamming distance of~$\C$. By definition, 
\begin{equation}\label{eq:dr(C)}
	d_r(\C)=n-{\rm max}\{N(H_r)\,:\,  H_r \in [ \C, r]\}
\end{equation}
where 
$N(H_r)= |\{i\,:\, 1\leq i\leq n,\ c_i=0\  {\rm for}\ \forall \ {\bm c}=(c_1,c_2,\dots,c_n)\in H_r \}|.$

We recall several basic properties of generalized Hamming weights that will be used throughout this paper. 
\begin{lemma}\cite[Corollary 1]{Wei1991}\label{rMDS}
	For an $[n,k]$ linear code $\mathcal{C}$, $d_r(\mathcal{C})$ satisfies
	\[d_r(\mathcal{C})\leq n-k+r.\]
\end{lemma}

 When $r=1$, this is the Singleton bound. In \cite{Wei1991}, the author gave  the definition on $r$-th rank MDS, or simply called an $r$-MDS code,  i.e.,  an $[n,k]$ linear code is called $r$-th rank MDS if $d_r(\C)=n-k+r$. 

\begin{lemma}\cite[Theorem 1]{Wei1991}\label{d1<d2<dots<dr} 
	For an $[n,k]$ linear code $\mathcal{C}$ with $k>0$, we have
	\[
	1\leq d_1(\mathcal{C})<d_2(\mathcal{C})<\cdots<d_k(\mathcal{C})\leq n.
	\]
\end{lemma}

If an $[n,k]$ linear code $\mathcal{C}$ is $r$-MDS, then it is also  a $(r+i)$-MDS code for $1\leq i\leq k-r$ by Lemma \ref{rMDS} and Lemma \ref{d1<d2<dots<dr}.  We have
\[
d_{r+i}(\mathcal{C})\ge d_r(\mathcal{C})+i= n-k+r+i.
\]
On the other hand, the generalized Singleton bound gives $d_{r+i}(\mathcal{C})\le n-k+r+i$. Hence $d_{r+i}(\mathcal{C})=n-k+r+i$, and therefore $\mathcal{C}$ is $(r+i)$-MDS.

\begin{lemma}\cite[Theorem 3]{Wei1991}\label{dual GHW of C}
	Let $\mathcal{C}$ be an $[n,k]$ linear code. Then
	\[\{d_r(\mathcal{C})\,:\,1\leq r\leq k\}=\{1,2,\dots,n\}\setminus \{n+1-d_r(\mathcal{C}^{\perp})\,:\,1\leq r\leq n-k\}.
	\]
\end{lemma}

\begin{lemma}\cite[Theorem 2]{Wei1991}\label{rank-characterization}
Let $\C$ be an $[n,k]$ linear code and let $H$ be a parity-check matrix of $\C$.
Then $d_r(\C)=\delta$ if and only if the following conditions hold:
\begin{enumerate}
    \item[(1)] any $\delta-1$ columns of $H$ have rank at least $\delta-r$;
    \item[(2)] there exist $\delta$ columns of $H$ with rank $\delta-r$.
\end{enumerate}
\end{lemma}

\subsection{Subcode Support Weight Distributions of Linear Codes}

Let $\C$ be an $[n,k]_q$ linear code. For $1\le r\le k$ and $0\le i\le n$, define
\[
A_i^{(r)}(\C)=\bigl|\{\mathcal{D}\in[\C,r]: |\supp(\mathcal{D})|=i\}\bigr|.
\]
The sequence $\bigl(A_0^{(r)}(\C),A_1^{(r)}(\C),\ldots,A_n^{(r)}(\C)\bigr)$ is called the $r$-subcode support weight distribution
($r$-SSWD) of $\C$. The subcode support weight distribution can be viewed as a
refinement of the generalized Hamming weight. Indeed,
\[
d_r(\C)=\min\bigl\{i:A_i^{(r)}(\C)\ne 0\bigr\}.
\]
Moreover, when $r=1$, every one-dimensional subcode contains $q-1$ nonzero codewords having the same support. Hence $A_i(\C)=(q-1)A_i^{(1)}(\C)$, where $A_i(\C)$ denotes the number of codewords of Hamming weight $i$ in $\C$. Thus, the $1$-SSWD is equivalent to the ordinary weight distribution.

For integers $0\le r\le k$, the $q$-ary Gaussian binomial coefficient is defined by
\[
\begin{bmatrix}
k\\ r
\end{bmatrix}_q=\frac{(q^k-1)(q^k-q)\cdots(q^k-q^{r-1})}{(q^r-1)(q^r-q)\cdots(q^r-q^{r-1})}.
\]
We set $\begin{bmatrix}
k\\ 0
\end{bmatrix}_q=1$ and $\begin{bmatrix}
k\\ r
\end{bmatrix}_q=0$ if $r<0$ or $r>k$. Since an $[n,k]_q$ linear code has exactly $\begin{bmatrix}k\\r\end{bmatrix}_q$ $r$-dimensional
subcodes, we have
\[
\sum_{i=0}^{n}A_i^{(r)}(C)=
\begin{bmatrix}
k\\r
\end{bmatrix}_q.
\]

\subsection{TGRS Codes}
Let \(\mathbb{F}_q[x]\) denote the polynomial ring over the finite field \(\mathbb{F}_q\). Take \(\boldsymbol{\alpha}=(\alpha_1,\dots,\alpha_n)\in\mathbb{F}_q^n\), where \(\alpha_1,\dots,\alpha_n\) are pairwise distinct elements of \(\mathbb{F}_q\), called evaluation points. Let \(\boldsymbol{v}=(v_1,\dots,v_n)\in(\mathbb{F}_q^*)^n\). We  define the evaluation map associated with \(\boldsymbol{\alpha}\) and \(\boldsymbol{v}\) as
\[
ev_{\boldsymbol{\alpha},\boldsymbol{v}}:\mathbb{F}_q[x]\to\mathbb{F}_q^n,\quad
f(x)\mapsto\bigl(v_1f(\alpha_1),\dots,v_nf(\alpha_n)\bigr).
\]
In this context, an \([n,k]\) generalized Reed-Solomon (RS) code, denoted \({\rm GRS}_k(\boldsymbol{\alpha},\boldsymbol{v})\), is the image of the restriction of \(ev_{\boldsymbol{\alpha},\boldsymbol{v}}\) to the set \(\mathbb{F}_q[x]_{<k}\), where \(\mathbb{F}_q[x]_{<k}\) consists of all polynomials of degree at most \(k-1\). When \(\boldsymbol{v}=\boldsymbol{1}\) (the all-one vector), \({\rm GRS}_k(\boldsymbol{\alpha},\boldsymbol{1})\) reduces to the standard RS code. Next, we introduce the definitions of twisted polynomials and TGRS codes.

\begin{definition}{\label{twistpolynomial}}\cite{Beelen2018}
	Let $k$ and $l$ be the positive integers with $1\leq l\leq k$. Let $\boldsymbol{h}=(h_1,h_2,\dots,h_{l})$ and $\boldsymbol{t}=(t_1,t_2,\dots,t_{l})$ be two vectors of length $l$ over ${\mathbb Z}$ such that  $h_i$ and $t_j$ are pairwise distinct with $0\leq h_i\leq k-1$ and $1\leq t_j\leq n-k$. Let $\boldsymbol{\eta}=(\eta_1,\eta_2,\dots,\eta_{l})\in ({\mathbb F_q^*})^{l}$. The set of $(\boldsymbol{t,h,\eta})$-twisted polynomials is defined as
	$$\mathcal{P}=\left\{\sum_{i=0}^{k-1}f_ix^i+\sum_{j=1}^{l}\eta_j f_{h_j}x^{k-1+t_j},\ f_i\in {\mathbb F_q}\right\}.$$
	 We call $\boldsymbol{h}=(h_1,h_2,\dots,h_{l})$ the hook and $\boldsymbol{t}=(t_1,t_2,\dots,t_{l})$ the twist.
	\end{definition}

\begin{definition}\cite{Beelen2018}
Let the symbols be defined as above. Twisted generalized Reed-Solomon (TGRS) code $\C(\boldsymbol{ \alpha, v})$ with length $n$ and dimension $k$ is defined as
\begin{equation}{\label{definition of code C}}
	\C(\boldsymbol{ \alpha, v})=\left\{ev_{\boldsymbol{\alpha,v} }(f(x)), \ f(x)\in \mathcal{P}\right\}.
\end{equation}
\end{definition}

\begin{remark}
	It is obvious that the set $\mathcal{P}$ of twisted polynomials  is a $k$-dimensional ${\mathbb F_q}$-linear subspace of ${\mathbb F_q}[x]$. For any $f=\sum_{i=0}^{k-1}f_ix^i+\sum_{j=1}^{l}\eta_j f_{h_j}x^{k-1+t_j} \in \mathcal{P}$, it is clear that $f$ can be expressed as $f=\sum_{i=0}^{k-1}f_ig_i$, where $f_i \in {\mathbb F}_q$,
	$$g_i=\begin{cases}
		x^i+ \eta_jx^{k-1+t_j}, & i=h_j,\\
	    x^i, & otherwise,
	\end{cases}
	$$ and $g_0,g_1,\dots,g_{k-1}$ are linearly independent over ${\mathbb F}_q$.
\end{remark}
Without loss of generality, we assume that $h_1<h_2< \cdots< h_{\ell}, \ t_1<t_2< \cdots <t_{\ell}.$   Obviously, the generator matrix of  code $\C(\bm \alpha,\bm v)$ is 
\begin{equation}
 G=\left( \begin{array}{cccc}
v_1 & v_2 & \cdots &  v_n \\
v_1\alpha_1 &  v_2\alpha_2 & \cdots & v_n\alpha_n \\
 \vdots & \vdots &  & \vdots\\
 v_1\alpha_1^{h_1-1} & v_2\alpha_2^{h_1-1} & \cdots & v_n\alpha_n^{h_1-1}\\
 v_1(\alpha_1^{h_1}+\eta_1 \alpha_1^{k-1+t_1}) & v_2(\alpha_2^{h_1}+\eta_1 \alpha_2^{k-1+t_1}) & \cdots & v_n(\alpha_n^{h_1}+\eta_1 \alpha_n^{k-1+t_1})\\
v_1\alpha_1^{h_1+1} & v_2\alpha_2^{h_1+1} & \cdots & v_n\alpha_n^{h_1+1}\\
\vdots & \vdots & & \vdots \\
v_1\alpha_1^{h_{\ell}-1} & v_2\alpha_2^{h_{\ell}-1} & \cdots & v_n\alpha_n^{h_{\ell}-1} \\
v_1(\alpha_1^{h_{\ell}}+\eta_{\ell}\alpha_1^{k-1+t_{\ell}}) & v_2(\alpha_2^{h_{\ell}}+\eta_{\ell}\alpha_2^{k-1+t_{\ell}}) & \cdots & v_n(\alpha_n^{h_{\ell}}+\eta_{\ell}\alpha_n^{k-1+t_{\ell}}) \\
v_1\alpha_1^{h_{\ell}+1} & v_2\alpha_2^{h_{\ell}+1}
 & \cdots & v_n\alpha_n^{h_{\ell}+1} \\
\vdots & \vdots & & \vdots \\
v_1\alpha_1^{k-1} & v_2\alpha_2^{k-1} & \cdots & v_n\alpha_n^{k-1}
  \end{array}
\right)
\end{equation}
\begin{remark}
Since all $v_i$ are nonzero for $i\in [n]$, $C(\boldsymbol{\alpha},\boldsymbol{v})$ is monomially equivalent to $C(\boldsymbol{\alpha},\boldsymbol{1})$. Such an equivalence preserves
the generalized Hamming weights and subcode support sizes. Therefore,
we may assume that $\boldsymbol{v}=\boldsymbol{1}$ throughout the
remainder of the paper.
\end{remark}

\subsection{Related Works and Motivation}

Generalized Hamming weights  are a natural generalization of the minimum Hamming distance. Since the systematic development of GHWs by Wei~\cite{Wei1991}, they have found applications
in wiretap channels of type II, resilient functions, trellis and branch complexity, code shortening, and erasure list decoding \cite{Wei1991,ChorEtAl1985,Forney1994,Guruswami2003}. The generalized Hamming weights or complete weight hierarchies have been
determined or investigated for many important classes of linear codes, including Hamming codes~\cite{Wei1991}, Reed-Muller
codes~\cite{HeijnenPellikaan1998}, Kasami
codes~\cite{HellesethKumar1995}, Melas and dual Melas
codes~\cite{GeerVlugt1994Melas}, BCH codes and their
duals~\cite{GeerVlugt1994BCH,DuursmaEtAl1996}, trace
codes~\cite{StichtenothVoss1994}, cyclic
codes~\cite{FengTzengWei1992,JanwaLal1997}, algebraic geometry codes~\cite{TsfasmanVladut1995}, and codes constructed from quadratic
forms~\cite{LiuZhengWang2023}.  Moreover, the relations between generalized Hamming weights and other
Hamming-type metrics have also attracted considerable attention. In particular, Shi et al.~\cite{ShiZhuHelleseth2023} investigated the connections among the  Hamming metric, the $b$-symbol metric, and the $r$-th generalized Hamming metric. Nevertheless, determining the complete weight hierarchy remains difficult for many families of linear codes.

The subcode support weight distribution (SSWD) provides a fine description of the support structure of a linear code by counting $r$-dimensional subcodes with each prescribed support size. In
particular, the $1$-SSWD is equivalent to the ordinary weight distribution , while the minimum support size appearing in the $r$-SSWD is the $r$-th generalized Hamming weight. Helleseth, Kl{\o}ve, and Mykkeltveit~\cite{HellesethKloveMykkeltveit1977} established a connection between subcode support weight distributions and the weight distributions of lifted codes, and Kl{\o}ve~\cite{Klove1992}
and \"Ozen and Tekin~\cite{OzenTekin2013} further studied structural relations and moments of support weight distributions. The SSWDs of several classes of linear codes have subsequently been determined,
including optimal codes constructed from down-sets
\cite{LuoLiu2023} and projective Solomon--Stiffler codes \cite{ShiLiHelleseth2024}. More recently, Pan et al.~\cite{PanChenLiuLiu2025} constructed several
families of distance-optimal few-SSW $r$-Griesmer codes and determined their subcode support weight distributions. Xie et al.~\cite{XieChenDingLi2026} established formulas relating the SSWDs of
projective linear codes to those of their simplex complementary codes and determined the SSWDs of several families of optimal codes. Pan et
al.~\cite{PanChenLiuZhangChen2026} further generalized the Pless power moments to subcode support weight distributions and applied them to
determine the SSWDs of linear codes with few subcode support weights.

Reed-Solomon  codes form one of the most important families of MDS codes. As a generalization of RS codes, twisted Reed-Solomon codes were introduced by Beelen et al.~\cite{Beelen2017}, and their structural properties were further studied in~\cite{Beelen2018}. Unlike RS codes, TRS codes are not necessarily MDS, and their MDS and NMDS properties have therefore attracted considerable attention. Sui et al.~\cite{Sui2022Zhu} established several characterizations and constructions of MDS and NMDS codes from TRS codes. Subsequently, TRS and TGRS codes with more
general twist structures were investigated in
\cite{Zhu2024,Gu2024,Hu2025}. These developments provide a natural motivation for studying finer parameters of TRS codes beyond the minimum distance, such as their generalized Hamming weights.

The generalized Hamming weights of codes related to TRS codes have also been considered recently. In particular, Yang et al. ~\cite{Yang2026} investigated the generalized Hamming weights of AMDS codes arising from TRS codes and obtained weight-hierarchy results for several $s$-gap AMDS cases.  However, the analysis is
restricted to particular almost-MDS and gap configurations and does not provide a uniform characterization of the complete weight hierarchy for general TRS codes with one twist.

This naturally raises the question of whether the entire weight hierarchy of a TRS code can be determined. More specifically, it is of interest to identify which parameters control the higher generalized Hamming weights and whether only a small number of
generalized Hamming weights of the dual code are sufficient to determine the whole hierarchy. For TRS codes with one twist, we show that all higher generalized Hamming weights are determined by
the minimum distance of the dual code. This viewpoint also leads naturally to the two-twist case, where the first two generalized Hamming weights of the dual code determine the higher weight hierarchy.

Once the generalized Hamming weight hierarchy is determined, a natural further question is to describe the support sizes of all $r$-dimensional subcodes. The $r$-th generalized Hamming weight captures only the minimum possible support size, whereas the $r$-subcode support weight distribution records the number of $r$-dimensional subcodes for every possible support size.
Thus, it is natural to further determine the SSWDs of TRS codes to obtain a finer description of their subcode support structure.

\subsection{Our Main Contributions and Organization of This Paper}
The main contributions of this paper are summarized as follows.
\begin{enumerate}
    \item [(1)] For TRS codes with one twist, we determine the complete generalized Hamming weight hierarchy given in Theorems \ref{thm:complete-single} and Proposition \ref{d1(C)}. Its higher generalized Hamming weights are determined by the minimum distance of the dual code, and the minimum distance of  its dual are characterized in Proposition \ref{dC perp}.

    \item [(2)] For TRS codes with two twists, we determine the generalized Hamming weight hierarchy in Theorem \ref{thm:complete} and Proposition \ref{prop:firsttwo d1 d2}, which is characterized by the first two generalized Hamming weights of the dual code, and given in Proposition \ref{prop:dualweights d1d2}.

    \item [(3)] For TRS codes with one twist, we derive an explicit formula for the subcode support weight distribution in Theorem \ref{thm:sswd}.
\end{enumerate}

The rest of this paper is organized as follows. Section~2 studies the generalized Hamming weights of TRS codes with one twist and gives their complete weight hierarchy. Section~3 extends these results to TRS codes with two twists and determines the corresponding weight hierarchy. Section~4 investigates the subcode support weight distributions of TRS codes with one twist and derives an explicit formula for them. Section 5 concludes this paper.

\section{GHW of the TRS Codes with a Single Twist}
In this section, we determine the weight hierarchy of  a class of TRS codes $\mathcal{C}_1({\bm \alpha},{\bm 1})$ with single twist $t$ and hook $h$. Let 
\begin{equation}\label{S}
	\mathcal{S}_1=\left\{\sum_{i=0}^{k-1}f_ix^i+\eta f_h x^{k-1+t},\ f_i\in {\mathbb F_q}\right\},
\end{equation}
where $0\leq h\leq k-1$  and $1\leq t\leq n-k$.  The TRS code $\mathcal{C}_1({\bm \alpha},{\bm 1})$ is defined as
\begin{equation}\label{def:C(a,1)}
	\mathcal{C}_1({\bm \alpha},{\bm 1})=\{(f(\alpha_1),f(\alpha_2),\dots, f(\alpha_n)),f\in \mathcal{S}_1\}.
\end{equation}

We next investigate the generalized Hamming weight hierarchy of $\mathcal{C}_1(\bm{\alpha},\bm{1})$. For convenience, we denote $\mathcal{C}_1(\bm{\alpha},\bm{1})$ by $\mathcal{C}_1$.  We first establish the possible
values of $d_r(\mathcal{C}_1)$ for each
$1\le r\le k$. This provides the basic structure of the weight hierarchy and will be refined in the subsequent results.

\begin{theorem}\label{rangeofd(C)}
Let $\mathcal{C}_1$ be the TRS code defined in~\eqref{def:C(a,1)} with $k\ge 2$. Then
\[
n-k-t+1\le d_1(\mathcal{C}_1)\le n-k+1,
\]
and
\[
n-k+r-1\le d_r(\mathcal{C}_1)\le n-k+r,
\ 2\le r\le k-1.
\]
Moreover, $d_k(\mathcal{C}_1)=n$.
\end{theorem}
\begin{proof}
	For the case $r=1$, since the code $\mathcal{C}_1$ is a subcode of ${\rm GRS}_{k+t}({\bm \alpha},{\bm 1})$, the minimum distance of $\mathcal{C}_1$ satisfies 
	$$d_1(\mathcal{C}_1)\geq d({\rm GRS}_{k+t}({\bm \alpha},{\bm 1}))=n-k-t+1,$$
	by the Singleton bound, we have 
	$$n-k-t+1\leq d_1(\mathcal{C}_1)\leq  n-k+1.$$

	For the case $2\leq r\leq k-1$, we define the evaluation map associated with ${\bm \alpha}$ as
	\[
ev_{\boldsymbol{\alpha}}:\mathcal{S}_1\to \mathcal{C},\quad
f(x)\mapsto\bigl(f(\alpha_1),f(\alpha_2),\dots,f(\alpha_n)\bigr),
	\]
where $\mathcal{S}_1$ is given in (\ref{S}), and it is clear that map  $ev_{\boldsymbol{\alpha}}$ is a linear isomorphism. Let 
\[
U={\rm span}\{x^i\,:\, 0\leq i\leq k-1,\ i\neq h\}, \ g(x)=x^h+\eta x^{k-1+t},
\]
we have $\mathcal{S}_1=U\bigoplus \left\langle g(x) \right\rangle $, and ${\rm dim}\, \mathcal{S}_1=k$.  For any $r$-dimensional subcode $\mathcal{D}$ of the code $\mathcal{C}$, there exists a $r$-dimensional subspace $W$ of $\mathcal{S}_1$ such that $W=ev_{\boldsymbol{\alpha}}^{-1}(\mathcal{D})$. From equation (\ref{eq:dr(C)}), we have
\begin{equation}\label{eq:dr(Calpha)}
	d_r(\mathcal{C}_1)=n- {\rm max}\{| N(\mathcal{D})|\,:\,\mathcal{D}\in [\mathcal{C}_1, r]\},
\end{equation}
where 
\[N(\mathcal{D})= \{i\,:\, 1\leq i\leq n,\ c_i=0\ {\rm for}\ \forall \ {\bm d}=(f(\alpha_1),f(\alpha_2),\dots,f(\alpha_n))\in \mathcal{D},\, f\in W \}.\]
Denote $N(\mathcal{D})=I$, for any $i\in I$, $f\in W$, we have $f(\alpha_i)=0$. Let $P_I(x)=\prod\limits_{i\in I}(x-\alpha_i)$, then $P_I(x)=\prod\limits_{i\in I}(x-\alpha_i)\mid f(x)$ for any $f(x)\in  W$, i.e., for any $u(x)\in W\cap U$, we have $P_I(x)=\prod\limits_{i\in I}(x-\alpha_i)\mid u(x)$, and 
\begin{equation}\label{eq:WcapU}
	W\cap U \subseteq \{u(x)\,:\, P_I(x)=\prod\limits_{i\in I}(x-\alpha_i)\mid u(x)\}.
\end{equation}
Let $V_I$ be the vector space consisting of all polynomials of degree less than $k$  which are divisible by $P_I(x)=\prod\limits_{i\in I}(x-\alpha_i)$, and
\[
V_I=\{v(x)\in {\mathbb{F}_q}[x]_{<k}\,:\, P_I(x)\mid  v(x)\}.
\]
For any $v(x)\in V_I$, there exists a polynomial $q(x)$ such that $v(x)=P_I(x)q(x)$. Since ${\rm deg}(v(x))\leq k-1$, we have ${\rm deg}(q(x))\leq k-|I|-1$. We consider the following map:
\[
\varphi :{\mathbb{F}_q}[x]_{<k-|I|} \to V_I,\quad
q(x) \mapsto P_I(x)q(x),
\]
it is easy to verify that $\varphi$ is a linear isomorphism. Hence, ${\rm dim}\,V_I={\rm dim}\,{\mathbb{F}_q}[x]_{<k-|I|}=k-|I|$, i.e., $V_I={\rm span}\{P_I(x),xP_I(x),\dots,x^{k-|I|-1}P_I(x)\}$. It is known that for any $u(x)\in U$, ${\rm deg}(u(x))<k$. So by equation (\ref{eq:WcapU}), we have 
\begin{equation}\label{eq:WcapUleq}
	W\cap U \subseteq V_I,\ {\rm dim}\, (W\cap U)\leq {\rm dim}\, V_I=k-|I|.
\end{equation} 
On the other side, 
\[{\rm dim}\,(W\cap U)={\rm dim}\,W+{\rm dim}\,U-{\rm dim}\,(W+U),\ {\rm dim}\,(W+U)\leq {\rm dim}\,\mathcal{S}_1=k,\]
then ${\rm dim}\,(W\cap U)\geq r+k-1-k=r-1$. From (\ref{eq:WcapUleq}), we have $|I|\leq k-r+1$. by the formula (\ref{eq:dr(Calpha)}), have $d_r(\mathcal{C}_1)\geq n-k+r-1$, and by Lemma \ref{rMDS}, we have 
\[n-k+r-1\leq d_r(\mathcal{C})\leq n-k+r.\]

For the case $r=k$, it can be shown that for each coordinate, there exists some codeword whose value at that coordinate is nonzero. We have the following subcases:
\begin{itemize}
	\item [(1)] $h\neq 0$. It is obviously ${\bm 1}\in \mathcal{C}_1$.
	\item [(2)] $h=0$, $\alpha_i\neq 0$ for $1\leq i\leq n$. Then there exists a codeword $(\alpha_1,\alpha_2,\dots,\alpha_n)\in \mathcal{C}_1$ such that every coordinate is nonzero.
	\item [(3)] $h=0$ and there exists an evaluation point $\alpha_i=0$, where $1\leq i\leq n$. Then $g(x)=1+\eta x^{k-1+t}\in \mathcal{S}_1$, i.e., $g(\alpha_i)=g(0)=1\neq 0$.
\end{itemize}
This completes the proof.
\end{proof}

To determine its exact value, we now give the following result on determining its minimum distance more explicitly.
\begin{proposition}\label{d1(C)}
	Let $\mathcal{C}_1$ be a TRS code defined in (\ref{def:C(a,1)}). Let $I\subseteq [n]$ and $P_I(x)=\prod\limits_{i\in I}(x-\alpha_i)=\sum\limits_{j=0}^{|I|}c_jx^{|I|-j}$. For $0\leq  \mu \leq t$, let $A_{\mu}(I)$ be a matrix given as follows
	\[
	A_{\mu}(I)=\left(\begin{array}{ccccc}
	c_1 & c_0  & 0 & \cdots & 0\\
	c_2 & c_1 & c_0 & \cdots & 0\\
	\vdots & \vdots & \vdots & & \vdots  \\
	c_{t-1} & c_{t-2} & c_{t-3} & \cdots & c_{\mu-1} \\
	1-\eta c_{k-1+t-h} & -\eta c_{k-1+t-h-1} & -\eta c_{k-1+t-h-2} & \cdots & -\eta c_{k-1+t-h-(t-\mu)}
\end{array}\right)_{t\times (t-\mu+1)},
	\]
	where $c_i=0$ if $i\notin \{0,1,2,\dots,|I|\}$. 
	Then $d_1(\mathcal{C}_1)=n-k+1-\mu$ if and only if there exist a $(k-1+\mu)$-subset $I_1$ such that ${\rm rank}\, A_{\mu}(I_1)\leq t-\mu$ and for any $(k+\mu)$-subset $I_2$, ${\rm rank}\, A_{\mu+1}(I_2)=t-\mu$.
\end{proposition}
\begin{proof}
	It is known that 
	\begin{align*}
			d_1(\mathcal{C}_1) & =\min\limits_{0\neq f\in \mathcal{S}_1} | \{i\,:\, 1\leq i\leq n,\ f(\alpha_i)\neq 0\}|\\
			& =n-\max\limits_{0\neq f\in \mathcal{S}_1} |\{i\,:\, 1\leq i\leq n,\ f(\alpha_i)=0\}|.
	\end{align*}
Then $d_1(\mathcal{C}_1)=n-k+1-\mu$ if and only if 
\begin{equation}\label{zeropointset of f}
	N=\max\limits_{0\neq f\in \mathcal{S}_1} | \{i\,:\, 1\leq i\leq n,\ f(\alpha_i)=0\} | =k-1+\mu.
\end{equation}
The equation (\ref{zeropointset of f}) holds if and only if there exists an $0\neq f\in \mathcal{S}_1$ with at least $k-1+\mu $ zeros, and there does not exist an $0\neq f\in \mathcal{S}_1$ with at least $k+\mu$ zeros. Assume that $f(x)\in \mathcal{S}_1$ has at least $k-1+\mu$ zeros, take a subset $I\subseteq [n]$ of the index set of these zeros with size exactly $k-1+\mu$. Then we have
\[
f(x)=P_I(x)q(x)=\prod\limits_{i\in I}(x-\alpha_i)q(x)=\sum_{j=0}^{k-1+\mu}c_jx^{k-1+\mu-j}q(x).
\]
Since ${\rm deg}(f(x))\leq k-1+t$, we have ${\rm deg}(q(x))\leq k-1+t-(k-1+\mu)=t-\mu$ and denote $q(x)=\sum\limits_{i=0}^{t-\mu}b_ix^{t-\mu-i}$, where $b_i$ are not all zero for $0\leq i\leq t-\mu$. Then $f(x)\in \mathcal{S}_1$ if and only if 
\begin{equation}\label{eq:fk}
	f_k=f_{k+1}=\cdots=f_{k-2+t}=0\ {\rm and}\ f_{k-1+t}=\eta f_h,
\end{equation}
where $f_j$, $j\in \{h,k,k+1,\dots,k-1+t\}$ denotes coefficient of the term $x^j$ in the polynomial $f(x)$. From equation (\ref{eq:fk}), we have
\[
\sum_{i=0}^{t-\mu}b_ic_{m-i}=0,\ 1\leq m\leq t-1,\ {\rm and }\ b_0=\eta \sum_{i=0}^{t-\mu}b_ic_{k-1+t-h-i}.
\]
In particular, 
\begin{equation}\label{eq:system equation}
	\left(\begin{array}{ccccc}
	c_1 & c_0  & 0 & \cdots & 0\\
	c_2 & c_1 & c_0 & \cdots & 0\\
	\vdots & \vdots & \vdots & & \vdots  \\
	c_{t-1} & c_{t-2} & c_{t-3} & \cdots & c_{\mu-1} \\
	1-\eta c_{k-1+t-h} & -\eta c_{k-1+t-h-1} & -\eta c_{k-1+t-h-2} & \cdots & -\eta c_{k-1+t-(t-\mu)}
\end{array}\right)
\left(\begin{array}{c}
	b_0\\
	b_1\\
	\vdots \\
	b_{t-\mu}
\end{array}\right)
={\bm 0},
\end{equation}
where  the index $i$ of $c_i$ ranges from $0$ to $k-1+\mu$, $0\leq \mu \leq t$, otherwise $c_i=0$. Let $A_{\mu}(I)$ denote the $t\times (t-\mu+1)$ coefficient matrix of the system of equations in (\ref{eq:system equation}). Then there exists an $0\neq f\in \mathcal{S}_1$ with at least $k-1+\mu$ zeros if and only if there exists a $(k-1+\mu)$-subset $I$ such that the system of equations (\ref{eq:system equation}) has nonzero solutions, i.e.,  ${\rm rank}\,A_{\mu}(I)\leq t-\mu$. Naturally, there does not exist an $0\neq f\in \mathcal{S}_1$ with at least $k+\mu$ zeros if and only if there does not exist $(k+\mu)$-subset $J$ such that ${\rm rank}\, A_{\mu+1}(J)\leq t-\mu-1$. Since $A_{\mu+1}(J)$ is a $t\times (t-\mu)$ matrix, it is also equivalent to  that for any $(k+\mu)$-subset $J$, ${\rm rank}\, A_{\mu+1}(J)=t-\mu$. This completes the proof.
\end{proof}

\begin{remark}
Proposition~\ref{d1(C)} is consistent with the known MDS criterion for the TRS codes $\mathcal{C}$ in
\cite[Lemma~2.6]{Huang2022}. More precisely, when $t=1$, the MDS condition in Proposition~\ref{d1(C)} is exactly equivalent to the condition
\[
\eta^{-1}\in\mathbb F_q^*\setminus S_{k,h},\ S_{k,h}=\left\{(-1)^{k-h}\sum_{\{j_1,\dots,j_{k-h}\}\subseteq I}\prod_{\ell=1}^{k-h} \alpha_{j_{\ell}}\,:\,I\subseteq[n],\ |I|=k\right\},
\]
given in \cite[Lemma~2.6]{Huang2022}. Indeed, an MDS code satisfies $d_1(\C)=n-k+1$, and hence $\mu=0$ in Proposition~\ref{d1(C)}. For $t=1$, the
condition involving a $(k-1)$-subset is automatic, whereas for
every $k$-subset $I\subseteq[n]$ the second condition becomes ${\rm rank}\, A_1(I)=1$. Since
\[
A_1(I)=
\begin{pmatrix}
1-\eta c_{k-h}
\end{pmatrix}.
\]
Therefore, ${\rm rank}\, A_1(I)=1$ if and only if $1-\eta c_{k-h}\neq0$, i.e.,
\[
\eta^{-1}\notin\left\{(-1)^{k-h}\sum_{\{j_1,\dots,j_{k-h}\}\subseteq I}\prod_{\ell=1}^{k-h} \alpha_{j_{\ell}}\,:\,I\subseteq[n],\ |I|=k
\right\}=S_{k,h},
\]
which is precisely the MDS criterion in \cite[Lemma~2.6]{Huang2022}.
Hence the known MDS characterization for the case $t=1$ is recovered as a special case of Proposition~\ref{d1(C)}.
\end{remark}

\begin{corollary}\label{cor:t2-hk1}
Let $\C_1$ be the TRS code defined in~\eqref{def:C(a,1)} with $t=2$ and $h=k-1$. For each subset $I\subseteq[n]$, define $\sigma_1(I)=\sum\limits_{i\in I}\alpha_i,\
\sigma_2(I)=\sum\limits_{{i,j\in I,\,i<j}}\alpha_i\alpha_j$, and
\[
\Delta(I)=\eta\bigl(\sigma_2(I)-\sigma_1(I)^2\bigr)-1.
\]
Then the minimum distance of $\C_1$ is characterized as follows.

\begin{enumerate}
\item[(1)]
$d_1(\C_1)=n-k+1$ if and only if $\Delta(I)\neq 0$ for every $k$-subset $I\subseteq[n]$.

\item[(2)]
$d_1(\C_1)=n-k$ if and only if there exists a $k$-subset $I\subseteq[n]$ such that $\Delta(I)=0$, and for every $(k+1)$-subset $J\subseteq[n]$, $\bigl(\sigma_1(J),\,1-\eta\sigma_2(J)\bigr)\neq(0,0)$.

\item[(3)]
$d_1(\C_1)=n-k-1$ if and only if there exists a $(k+1)$-subset $J\subseteq[n]$ such that $\sigma_1(J)=0\ \text{and}\  \eta\sigma_2(J)=1$.
\end{enumerate}
\end{corollary}

\begin{example}
Let $\mathcal{C}_1$ be a TRS code defined in (\ref{def:C(a,1)}). Let $q=7$, $n=7$, $k=4$, $t=2$, $h=1$, and $\eta=1$, and take $\bm{\alpha}=(0,1,2,3,4,5,6)$. Take $\mu=2$. Let $I=\{1,3,4,5,6\}$, which corresponds to the evaluation points $\{0,2,3,4,5\}$. Thus
\[
A_2(I)=
\begin{pmatrix}
c_1\\
1-\eta c_{k-1+t-h}
\end{pmatrix}=
\begin{pmatrix}
0\\
0
\end{pmatrix},
\]
and hence ${\rm rank}\, A_2(I)=0=t-\mu$. Therefore the first condition in Proposition \ref{d1(C)} is satisfied. Moreover, every nonzero polynomial in $\mathcal{S}_1$ has degree at most $k-1+t=5$, and hence cannot have $k+\mu=6$ distinct zeros. Thus the second condition is automatic for $\mu=t$. Consequently, with the help of SageMath, 
\[
d_1(\mathcal{C}_1)=2=n-k+1-\mu,
\]
it is consistent  with Proposition \ref{d1(C)}.
\end{example}

We next consider the higher generalized Hamming weights of $\mathcal{C}_1$ and further refine the weight
hierarchy obtained in Theorem \ref{rangeofd(C)}.

\begin{theorem}\label{border of dr}
	Let $\mathcal{C}_1$ be a non-MDS TRS code defined in (\ref{def:C(a,1)}). Then $d_r(\mathcal{C}_1)$ satisfies 
	\begin{equation*}
		d_r(\mathcal{C}_1)=\begin{array}{cc}
			\begin{cases}
				n-k+r-1, & 2\leq r<k+2-d(\mathcal{C}_1^\perp)\\
				n-k+r, & k+2-d(\mathcal{C}_1^\perp)\leq r\leq k.
			\end{cases}
		\end{array}
	\end{equation*}
\end{theorem}
\begin{proof}
	Define
\[
r_0=\min\left\{r: 2\le r\le k,\;d_r(\mathcal{C}_1)=n-k+r\right\}.
\]
	By Lemma \ref{d1<d2<dots<dr}, we have 
	\[
	d_2(\mathcal{C}_1)<d_3(\mathcal{C}_1)<\cdots<d_k(\mathcal{C}_1)=n.
	\]
	It is known that $d_{r_0}(\mathcal{C}_1)=n-k+r_0$,  then  we have 
	\[
	d_{r_0+1}(\mathcal{C}_1)\geq d_{r_0}(\mathcal{C}_1)+1=n-k+r_0+1,
	\]
	furthermore, by Singleton bound,  $d_{r_0+1}(\mathcal{C}_1)\leq n-k+r_0+1$. Therefore, we have $d_{r_0+1}=n-k+r_0+1$. By mathematical induction, for any $r\geq r_0$, we have $d_r(\mathcal{C}_1)=n-k+r$. Similarly, we consider the case for $2\leq r<r_0$. It is known that \[
	d_r(\mathcal{C}_1)\in \{n-k+r-1,n-k+r\},\ 2\leq r\leq k-1,
	\]
	 and $d_{r_0-1}(\mathcal{C}_1)\neq n-k+r_0-1$, then we have $d_{r_0-1}(\mathcal{C}_1)=n-k+r_0-2$. Therefore, we have $d_{r_0-2}(\mathcal{C}_1)\leq n-k+r_0-3$, i.e., $d_{r_0-2}(\mathcal{C}_1)=n-k+r_0-3$.  By mathematical induction, for any $2\leq r< r_0$, we have $d_r(\mathcal{C}_1)=n-k+r-1$.

	Next, we prove that $r_0=k+2-d(\mathcal{C}_1^\perp)$. for $2\le r<r_0$, we have
\[
n-k+r_0-1\notin \{d_2(\mathcal{C}_1),d_3(\mathcal{C}_1),\dots,d_k(\mathcal{C}_1)\}
\]
Consequently, $n-k+r_0-1$ is the largest integer in $\{1,2,\ldots,n\}$ which is not a generalized Hamming weight of $\mathcal{C}_1$. By Lemma \ref{dual GHW of C}, we have
\[
\{d_i(\mathcal{C}_1):1\le i\le k\}
=\{1,2,\ldots,n\}\setminus \{n+1-d_j(\mathcal{C}_1^\perp):1\le j\le n-k\}.
\]
Since the generalized Hamming weights of $\mathcal{C}_1^\perp$ are strictly
increasing by Lemma \ref{d1<d2<dots<dr}, i.e., we have
\[
d_1(\mathcal{C}_1^\perp)<d_2(\mathcal{C}_1^\perp)<\cdots<d_{n-k}(\mathcal{C}_1^\perp),
\]
Hence, the largest integer which is not a
generalized Hamming weight of $\mathcal{C}_1$ is $n+1-d_1(\mathcal{C}_1^\perp)$.
Therefore, we have
\[
n-k+r_0-1=n+1-d_1(\mathcal{C}_1^\perp),
\]
and  $r_0=k+2-d_1(\mathcal{C}_1^\perp)$. This completes the proof.
\end{proof}

\begin{example}
Let $\mathcal{C}_1$ be a TRS code defined in (\ref{def:C(a,1)}). Let $q=7$, $n=7$, $k=5$, $t=1$, $h=2$, $\eta=1$, and take $\bm{\alpha}=(0,1,2,3,4,5,6)\in\mathbb F_7^7$. With the help of SageMath, we have
\[
d_1(\C_1)=2,\, d_2(\C_1)=3,\,d_3(\C_1)=4\,,d_4(\C_1)=6,\,d_5(\C_1)=7.
\]
Moreover, the minimum distance of  dual code is $d_1(\C^\perp)=3$. This is consistent with Theorem \ref{border of dr}.
\end{example}

To further determine the generalized Hamming weight hierarchy of $\mathcal{C}_1$, it remains to determine the minimum distance of its dual code $\mathcal{C}_1^\perp$. For this purpose, we first recall a useful lemma, which will be used in the subsequent characterization of
$d(\mathcal{C}_1^\perp)$.

\begin{lemma}\cite[Lemma 4]{Hu2025}\label{linear representation pf V_r}
	Let $\alpha_1,\alpha_2,\dots,\alpha_n$ be distinct elements of ${\mathbb F_q}$, and $m(x)=\prod\limits_{i=1}^n(x-\alpha_i)=\sum\limits_{j=0}^nc_jx^{n-j}$. For $t\geq 1$ and $0\leq i\leq n-1$,  $a_{n,i}^{(t)}\in {\mathbb F_q}$ satisfy the following system:
	$$(\alpha_{1}^{n-1+t},\alpha_{2}^{n-1+t},\dots,\alpha_{n}^{n-1+t})=(a_{0}^{(t)}, a_{1}^{(t)},\dots, a_{n-1}^{(t)})\left(\begin{array}{cccc}
		1 & 1 & \cdots & 1\\
		\alpha_{1} & \alpha_{2} & \cdots & \alpha_{n} \\
		\vdots & \vdots & \vdots & \vdots \\
		\alpha_{1}^{n-1} & \alpha_{2}^{n-1} & \cdots & \alpha_{n}^{n-1}
	\end{array}
	\right).$$
	Then we have
	\begin{equation}\label{sum of a_r,i^t}
		a_{s}^{(t)}=-\sum\limits_{j=0}^{min\{s,t-1\}}c_{n+j-s} e_{n,t-1-j} \ {\rm for}\ 0\leq s\leq n-1,
	\end{equation}
	where $e_{n,0}=1$ and $e_{n,i}=-\sum\limits_{j=0}^{i-1}e_{n,j}c_{i-j}$ for all $1\leq i\leq t-1$.
\end{lemma}

Using the above result, we now characterize the
minimum distance of the dual code $\mathcal{C}_1^\perp$.

\begin{proposition}\label{dC perp}
Let $\mathcal{C}_1$ be a non-MDS TRS code defined in
(\ref{def:C(a,1)}). For an $\nu$-subset $I\subseteq[n]$, where
$1\le \nu \le k$, let $P_I(x)=\prod\limits_{i\in I}(x-\alpha_i)=\sum\limits_{j=0}^{\nu}c_jx^{\nu-j}$. 
Then
\[
d_1(C_1^\perp)=\min\left\{\nu\,:\,
\begin{array}{l}
\exists\, I\subseteq[n],\ |I|=\nu,\ \text{such that}\\[1mm]
c_{\nu-h}=c_{\nu-h+1}=\cdots=c_{k-h-1}=0,\\[1mm]
\displaystyle
\eta\sum_{\ell=0}^{\min\{\nu-k+h,t-1\}}
c_{k-h+\ell}e_{t-1-\ell}=1
\end{array}
\right\}.
\]
where $e_{0}=1$ and $e_{i}=-\sum\limits_{j=0}^{i-1}e_{j}c_{i-j}$ for all $1\leq i\leq t-1$.
\end{proposition}

\begin{proof}
Let $G$ is a generator matrix of $\mathcal{C}_1$, it is also a parity-check matrix of $\mathcal{C}_1^\perp$. Hence $d(\mathcal{C}_1^\perp)$ is equal to the minimum number of linearly dependent columns of $G$. Fix a $\nu$-subset $I=\{i_1,i_2,\dots,i_{\nu}\}\subseteq[n]$, and let $G_I$ denote the submatrix of $G$ consisting of the columns indexed by $I$. For $0\le r\le \nu-1$, denote ${\bm V}_r=(\alpha_{i_1}^{r},\alpha_{i_2}^{r},\ldots,\alpha_{i_{\nu}}^{r})$. We first characterize when the columns of $G_I$ are linearly
dependent.

If $\nu \le h$, then the rows of $G_I$ corresponding to ${\bm V}_0,{\bm V}_1,\dots,{\bm V}_{\nu-1}$ form an $\nu \times \nu$ Vandermonde matrix. Since the evaluation points
$\{\alpha_i:i\in I\}$ are pairwise distinct, its determinant is nonzero. Hence ${\rm rank}\,(G_I)=\nu$, and therefore the columns of $G_I$ are linearly independent. Consequently, a linearly dependent set of $\nu$ columns can occur only when $h<\nu \le k$.

Now suppose that $h<\nu \le k$. Let 
\[
P_I(x)=\prod_{i\in I}(x-\alpha_i)
      =\sum_{j=0}^{\nu}c_jx^{\nu-j},
\ c_0=1.
\]
Since the evaluation points $\alpha_{i_1},\ldots,\alpha_{i_{\nu}}$ are pairwise distinct, the vectors ${\bm V}_0,{\bm V}_1,\dots,{\bm V}_{\nu-1}$ are linearly independent. Therefore, ${\bm V}_0,\dots,{\bm V}_{h-1}, {\bm V}_{h+1},\ldots,{\bm V}_{\nu-1}$ are also linearly independent. Hence the $\nu$ columns of $G_I$ are linearly dependent if and only if every remaining row of $G_I$ belongs to the space 
\[
W={\rm span}\{{\bm V}_0,\ldots,{\bm V}_{h-1},{\bm V}_{h+1},\ldots,{\bm V}_{\nu-1}\}.
\]
The remaining rows correspond to
\[
{\bm V}_{\nu},{\bm V}_{\nu+1},\dots,{\bm V}_{k-1},\ {\rm and}\ {\bm V}_h+\eta {\bm V}_{k-1+t}.
\]
By Lemma \ref{linear representation pf V_r}, for every $u\geq 1$, we have
\begin{equation}\label{eq:Vs-1+u}
	{\bm V}_{\nu-1+u}=-\sum_{r=0}^{\nu-1}\left(\sum_{j=0}^{\min\{r,u-1\}}c_{\nu+j-r}e_{u-1-j}\right){\bm V}_r,
\end{equation}
where $e_0=1$, and $e_i=-\sum\limits_{j=0}^{i-1}e_jc_{i-j}$ for  $i\geq 1$. We first consider  rows ${\bm V}_{\nu},{\bm V}_{\nu+1},\dots,{\bm V}_{k-1}$. For $1\le u\le k-\nu$, the row ${\bm V}_{\nu-1+u}$, belongs to the space $W$ if and only if 
\[
{\rm rank}({\bm V}_0^T,\dots,{\bm V}_{h-1}^T,{\bm V}_{h+1}^T,\dots,{\bm V}_{\nu-1}^T,{\bm V}_{\nu-1+u}^T)=\nu-1.
\]
Hence, by equation (\ref{eq:Vs-1+u}), we have 
\begin{equation}\label{eq:sum of ce}
	\sum_{j=0}^{\min\{h,u-1\}}
c_{\nu+j-h}e_{u-1-j}=0.
\end{equation}
For $u=1$, equation (\ref{eq:sum of ce}) gives $c_{\nu-h}e_0=0$. Since $e_0=1$, we obtain $c_{\nu-h}=0$. For $u=2$, equation (\ref{eq:sum of ce}) gives $c_{\nu-h}e_1+c_{\nu-h+1}e_0=0$, then we have $c_{\nu-h+1}=0$. Proceeding inductively, we have $c_{\nu-h+u-1}=0$. Therefore, the row ${\bm V}_{\nu-1+u}$, $1\leq u\leq k-\nu$ belongs to the space $W$  if and only if
\[
c_{\nu-h}=c_{\nu-h+1}=\cdots=c_{k-h-1}=0.
\]
It remains to consider the twisted row ${\bm V}_h+\eta {\bm V}_{k-1+t}$. By Lemma \ref{linear representation pf V_r}, we have 
\[
{\bm V}_{k-1+t}=-\sum_{r=0}^{\nu-1}\left(\sum_{j=0}^{\min\{r,k+t-\nu-1\}}c_{\nu+j-r}e_{k+t-\nu-1-j}\right){\bm V}_r.
\]
Consequently,
\begin{align*}
{\bm V}_h+\eta{\bm V}_{k-1+t}
=&-\eta\sum_{\substack{0\le r\le \nu-1\\r\ne h}}\left(\sum_{j=0}^{\min\{r,k+t-\nu-1\}}c_{\nu+j-r}e_{k+t-\nu-1-j}\right){\bm V}_r\\
&+\left(1-\eta\sum_{j=0}^{\min\{h,k+t-\nu-1\}}c_{\nu+j-h}e_{k+t-\nu-1-j}\right){\bm V}_h.
\end{align*}
Similarly, the row ${\bm V}_h+\eta{\bm V}_{k-1+t}$, belongs to the space $W$ if and only if 
\begin{equation}\label{eq:Vh+Vk-1+t}
	\eta\sum_{j=0}^{\min\{h,k+t-\nu-1\}}c_{\nu+j-h}e_{k+t-\nu-1-j}=1.
\end{equation}
Since $c_{\nu-h}=c_{\nu-h+1}=\cdots=c_{k-h-1}=0$. Thus equation (\ref{eq:Vh+Vk-1+t}) is equivalent to
\[
\eta
\sum_{\ell=0}^{\min\{\nu-k+h,t-1\}}
c_{k-h+\ell}e_{t-1-\ell}
=1.
\]
This completes the proof.
\end{proof}

\begin{remark}
Let $\mathcal{C}_1$ be a TRS code defined in
(\ref{def:C(a,1)}). For $t=1$ and $h\in\{0,k-1\}$, Propositions~\ref{d1(C)} and~\ref{dC perp} also recover the known 
characterization for the TRS codes to be NMDS codes in \cite[Lemma~2.6]{Huang2022}. Recall that
\[
S_{k,h}=\left\{(-1)^{k-h}\sum_{\{j_1,\dots,j_{k-h}\}\subseteq I}\prod_{\ell=1}^{k-h} \alpha_{j_{\ell}}\,:\,I\subseteq[n],\ |I|=k\right\}.
\]
Proposition~\ref{d1(C)} shows that $d_1(\C)=n-k$
if and only if  $\eta^{-1}\in S_{k,h}$. We next consider its dual code. For $t=1$, we have $e_0=1$, and the second condition in Proposition~\ref{dC perp} reduces to $\eta c_{k-h}=1$. If $h=0$, this condition can occur only for $\nu=k$, and hence $d_1(\C^\perp)=k$ if and only if $\eta c_k=1$, this is equivalent to
\[
\eta^{-1}\in S_{k,0}.
\]
Similarly, if $h=k-1$, $d_1(\C^\perp)=k$ if and only if 
\[
\eta^{-1}\in S_{k,k-1}.
\]
Therefore, $\C_1$ is NMDS if and only if $\eta^{-1}\in S_{k,h}$, which is exactly the NMDS characterization in \cite[Lemma~2.6]{Huang2022}.
\end{remark}

\begin{corollary}\label{cor:t2-dual}
Let $\C_1$ be the TRS code defined in~\eqref{def:C(a,1)} with $t=2$ and $h=k-1$. If $\C_1$ is non-MDS, then
\[
d_1(\C_1^\perp)=k.
\]
\end{corollary}

\begin{proof}
By Proposition \ref{dC perp}, $d_1(\C_1^\perp)=\nu$ only if
\[
c_{\nu-k+1}=c_{\nu-k+2}=\cdots=c_0=0.
\]
For $\nu<k$, this condition contains $c_0=0$, which is impossible
since $c_0=1$. Hence, we have $\nu\geq k$. Since $\C_1$ is non-MDS, we have $\nu\leq k$, therefore, we must have $d_1(\C_1^\perp)=k$.
\end{proof}

\begin{example}
Let $\mathcal{C}_1$ be a non-MDS TRS code defined in
(\ref{def:C(a,1)}). Let $q=11$, $n=9$, $k=4$, $t=2$, $h=1$, $\eta=6$, and take $\bm{\alpha}=(0,1,2,4,5,7,8,9,10)\in\mathbb F_{11}^{9}$.  We have its  generator matrix is 
\[
G=
\begin{pmatrix}
1&1&1&1&1&1&1&1&1\\
0&7&7&10&0&1&2&4&4\\
0&1&4&5&3&5&9&4&1\\
0&1&8&9&4&2&6&3&10
\end{pmatrix}.
\]
Then ${\rm rank}\, G_I=2$ for any $2$-subset $I$, and there exists the columns indexed by $J=\{2,4,7\}$ are linearly dependent. Therefore, $d_1(\C_1^\perp)=3$. We now verify this value by Proposition \ref{dC perp}. The subset $J=\{2,4,7\}$ corresponds to the polynomial
\[
P_I(x)=(x-1)(x-4)(x-8)=x^3+9x^2+1
\]
Hence, for $\nu=3$, $k=4$, and $h=1$, the first condition in
Proposition \ref{dC perp} is
\[
c_{\nu-h}=c_{k-h-1}=c_2=0,
\]
which is satisfied. For the second condition, 
\[
\eta\sum_{\ell=0}^{0}c_{k-h+\ell}e_{t-1-\ell}=6c_3e_1
=6\cdot1\cdot2=12=1.
\]
Thus the second condition is also satisfied.
\end{example}

With the minimum distances of $\mathcal{C}_1$ and its dual characterized above, Theorem \ref{border of dr} now yields a complete description of the generalized Hamming weight hierarchy of $\mathcal{C}_1$.

\begin{theorem}\label{thm:complete-single}
Let $\mathcal{C}_1$ be the  TRS code defined in~(\ref{def:C(a,1)}). If $\mathcal{C}_1$ is a MDS code, then
\[
d_r\bigl(\mathcal{C}_1\bigr)=
n-k+r,\ 1\le r\le k.
\]
Suppose that $\mathcal{C}_1$ is non-MDS. Then the generalized Hamming weight
hierarchy of $\mathcal{C}_1$ is given by
\[
d_r\bigl(\mathcal{C}_1\bigr)=
\begin{cases}
n-k+1-\mu,
& r=1,\\[1mm]
n-k+r-1,
& 2\le r<k+2-\nu,\\[1mm]
n-k+r,
& k+2-\nu
\le r\le k,
\end{cases}
\]
where $\mu$ and $\nu$ are determined by Proposition \ref{d1(C)} and Proposition \ref{dC perp}.
\end{theorem}

\begin{corollary}
Let $\C_1$ be a non-MDS TRS code defined in~(\ref{def:C(a,1)}). Then $d_1(\C_1^\perp)\ge \max\{h+1,k-h\}$. Consequently, $d_r(\C_1)=n-k+r$ for every $r\ge \min\{h+1,k-h\}+1$.
\end{corollary}

\begin{proof}
By Proposition \ref{dC perp}, a linearly dependent set of $\nu$ columns can occur only when $\nu>h$. Moreover, the second condition in Proposition \ref{dC perp} requires $\nu-k+h\ge 0$, and hence $\nu \ge k-h$. Therefore,
\[
d_1(\C_1^\perp)\ge\max\{h+1,k-h\}.
\]
The characterization  on the $d_r(\C_1)=n-k+r$ follows
immediately from Theorem \ref{thm:complete-single}.
\end{proof}

\begin{corollary}\label{cor:r-MDS-C1}
Let $\C_1$ be a non-MDS TRS code defined in~\eqref{def:C(a,1)}. Then, for $2\le r\le k$, $\C_1$ is $r$-MDS if and only if
\[
d_1(C_1^\perp)\ge k+2-r.
\]
\end{corollary}

\begin{remark}
The weight-hierarchy classifications of the  $s$-gap AMDS codes obtained in \cite{Yang2026} can be recovered uniformly from Theorem~\ref{thm:complete-single}. Let $G$ be a generator matrix of a  TRS code $\C$ given in Theorem 4.8 of \cite{Yang2026}, and assume that any $k-1$ columns of $G$ are linearly independent.  For each
$(k-1)$-subset $J\subseteq[n]$, let $N(J)$ denote the number of remaining columns contained in the span of the columns indexed by $J$, and put
\[
s=\max_{J\subseteq[n],\, |J|=k-1}N(J).
\]
Consequently, there exists a nonzero polynomial in the generator polynomial space $\mathcal{S}_1$ of $\C$ vanishing at $k-1+s$ evaluation points. Moreover, there exists no nonzero polynomial in the generator polynomial space $\mathcal{S}_1$ can vanish at more than $k-1+s$ evaluation points. Therefore,
\[
d_1(\C)=n-k+1-s.
\]
Thus, in the notation of Theorem~\ref{thm:complete-single}, we have $\mu=s$. Moreover, since any $k-1$ columns of $G$ are linearly independent and some $k$ columns are linearly dependent, we have $d_1(\C^\perp)=k$. It follows from Theorem~\ref{thm:complete-single} that
\[
\bigl(d_1(\C),\ldots,d_k(\C)\bigr)=\bigl(n-k+1-s,n-k+2,\ldots,n\bigr).
\]
By Lemma \ref{dual GHW of C}, we have
\[
\begin{aligned}
\{d_j(\C^\perp):1\le j\le n-k\}
&=
\{1,\ldots,n\}
 \setminus
 \{n+1-d_r(\C):1\le r\le k\}\\
&=
\{k,k+1,\ldots,n\}\setminus\{k+s\}.
\end{aligned}
\]
Thus $\C^\perp$ is an $s$-gap AMDS code.

In particular, for the case considered in
\cite[Theorem~4.8]{Yang2026}, the cases $s=1,\, s=2,\, s=3$ give, respectively, an NMDS code, a $2$-gap AMDS code, and a $3$-gap AMDS code. The case $s=0$ corresponds to an MDS code.
Hence the weight-hierarchy classification in
\cite[Theorem~4.8]{Yang2026} is recovered as a special case of Theorem~\ref{thm:complete-single} together with Lemma \ref{dual GHW of C}.
\end{remark}

\begin{table}[htbp]
\centering
\caption{Generalized Hamming weight hierarchies of single-twist TRS codes
}
\label{tab:single-twist-examples}
\resizebox{\textwidth}{!}{
\begin{tabular}{ccccccccccc}
\toprule
$q$ & $n$ & $k$ & $t$ & $h$ & $\eta$
& $\mu$ & $\nu$ &   $\boldsymbol{\alpha}$ & $(d_1,d_2,\dots,d_k)$\\
\midrule
$7$ & $5$ & $3$ & $1$ & $1$ & $3$
& $0$ & $4$  &  $(1,2,3,5,6)$ & $(3,4,5)$\\

$7$ & $5$ & $3$ & $1$ & $1$ & $3$
& $1$ & $2$  &  $(0,1,2,3,4)$ & $(2,3,5)$\\

$7$ & $5$ & $3$ & $1$ & $1$ & $3$
& $1$ & $3$  &  $(0,1,2,3,5)$ & $(2,4,5)$\\

7 & 6 & 4 & 2 & 3 & 1 & 1 & 4
& $(0,1,2,3,4,6)$
& $(2,4,5,6)$ \\

7 & 6 & 4 & 2 & 3 & 1 & 2 & 4
& $(0,1,2,3,4,5)$
& $(1,4,5,6)$ \\

$11$ & $9$ & $4$ & $2$ & $1$ & $6$
& $1$ & $4$  &  $(0,1,2,3,4,5,6,7,9)$ & $(5,7,8,9)$\\

$11$ & $9$ & $4$ & $2$ & $1$ & $6$
& $2$ & $3$  &  $(0,1,2,3,4,5,6,7,8)$ & $(4,6,8,9)$\\

$11$ & $8$ & $5$ & $3$ & $0$ & $1$ & $0$ & $6$
& $(0,1,2,3,4,6,9,10)$
& $(4,5,6,7,8)$ \\

$11$ & $8$ & $5$ & $3$ & $0$ & $1$ & $1$ & $5$
& $(0,1,2,3,4,5,6,7)$
& $(3,5,6,7,8)$ \\

$11$ & $8 $& $5$ & $3$ & $0$ &$ 1$ & $2$ & $5$
& $(0,1,2,3,4,5,6,8)$
& $(2,5,6,7,8)$ \\

$13$ & $10$ & $5$ & $1$ & $1$ & $7$ & $1$ & $4$ & $(0,1,2,3,4,5,6,7,8,9)$ & $(5,6,8,9,10)$\\

$13$ & $10$ & $5$ & $1$ & $1$ & $7$ & $1$ & $5$ & $(0,1,2,3,4,5,6,8,9,10)$ & $(5,7,8,9,10)$\\

13 & 11 & 7 & 2 & 1 & 1 & 1 & 6
& $(0,1,2,3,4,5,6,7,8,9,11)$
& $(4,5,7,8,9,10,11)$ \\

13 & 11 & 7 & 2 & 1 & 1 & 1 & 7
& $(0,1,2,3,4,5,6,7,8,9,10)$
& $(4,6,7,8,9,10,11)$ \\

13 & 11 & 7 & 2 & 1 & 1 & 2 & 6
& $(0,1,2,3,4,5,6,7,8,9,12)$
& $(3,5,7,8,9,10,11)$ \\

13 & 11 & 7 & 2 & 1 & 1 & 2 & 7
& $(0,1,2,3,4,5,6,7,8,10,12)$
& $(3,6,7,8,9,10,11)$ \\

$17$ & $14$ & $7$ & $1$ & $1$ & $3$ & $1$ & $6$ & $(0,1,2,3,4,5,6,7,8,9,10,11,12,13)$ & $(7,8,10,11,12,13,14)$\\

$19$ & $15$ & $7$ & $1$ & $0$ & $3$ & $1$ & $7$ & $(0,1,2,3,4,5,6,7,8,9,10,11,12,13,14)$ & $(8,10,11,12,13,14,15)$\\
\bottomrule
\end{tabular}
}
\end{table}

\section{GHW of the TRS Codes with two Twists}
In this section, we determine the weight hierarchy of  a class of TRS codes $\mathcal{C}_2({\bm \alpha},{\bm 1})$ with two twists ${\bm t}=(t_1,t_2)$ and hooks ${\bm h}=(h_1,h_2)$. Let 
\begin{equation}\label{eq:S}
	\mathcal{S}_2=\left\{\sum_{i=0}^{k-1}f_ix^i+\eta_1 f_{h_1} x^{k-1+t_1}+\eta_2 f_{h_2} x^{k-1+t_2},\ f_i\in {\mathbb F_q}\right\},
\end{equation}
where $0\leq h_1<h_2\leq  k-1$  and $1\leq t_1<t_2\leq n-k$.  The TRS code $\mathcal{C}_2({\bm \alpha},{\bm 1})$ is defined as
\begin{equation}\label{def:C(a,1)2}
	\mathcal{C}_2({\bm \alpha},{\bm 1})=\{(f(\alpha_1),f(\alpha_2),\dots, f(\alpha_n)),f\in \mathcal{S}_2\}.
\end{equation}

We next investigate the generalized Hamming weight hierarchy of $\mathcal{C}_2(\bm{\alpha},\bm{1})$. For convenience, we denote $\mathcal{C}_2(\bm{\alpha},\bm{1})$ by $\mathcal{C}_2$.  We first establish the possible
values of $d_r(\mathcal{C}_2)$ for each
$1\le r\le k$. This provides the basic structure of the weight hierarchy and will be refined in the subsequent results.

\begin{theorem}\label{thm:range}
Let $\C_2$ be a TRS code defined in (\ref{def:C(a,1)2}). Then we have
\[
n-k-t_2+1\le d_1(\C_2)\le n-k+1,\ n-k-t_1+1\le d_2(\C_2)\le n-k+2,
\]
and for $3\leq r\leq k$, we have
\[
n-k+r-2\le d_r(\C_2)\le n-k+r,
\]
Moreover, if $k\ge3$, then $d_k(\C_2)=n$.
\end{theorem}

\begin{proof}
The upper bounds follow immediately from the generalized Singleton bound. For an $r$-dimensional subcode $D$ of $\C_2$, let $W$ be the
corresponding $r$-dimensional subspace of $\mathcal S_2$. Define
\[
I(W)
=
\left\{
i\in[n]:
f(\alpha_i)=0
\text{ for every }f\in W
\right\}.
\]
Thus $I(W)$ is the set of common zero positions of all codewords in $D$. From equation (\ref{eq:dr(C)}), we have
\begin{equation}\label{eq:dr-common-zero}
d_r(\C_2)=n-\max \{ | I(W)|\,:\, W\leq \mathcal{S}_2,\ {\rm dim}\,W=r\}.
\end{equation}

For $r=1$, every nonzero polynomial $f\in\mathcal S_2$ satisfies $\deg f\le k-1+t_2$. Hence $f$ has at most $k-1+t_2$ zeros among the evaluation points, and
therefore
\[
d_1(\C_2)\ge n-k-t_2+1.
\]

We next consider $r=2$. Let $W$ be an arbitrary two-dimensional subspace of $\mathcal S_2$. Consider the linear map
\[
\varphi:W\longrightarrow\F_q,
\qquad
f\longmapsto f_{h_2},
\]
where $f_{h_2}$ denotes the coefficient of $x^{h_2}$ in $f$. Since $\dim W=2$ and $\dim\F_q=1$, then we have
\[
\dim\,\ker\varphi=\dim\, W-\dim\,\operatorname{Im}\varphi \ge 2-1=1.
\]
Therefore, there exists a nonzero polynomial $f\in W$
such that $\varphi(f)=f_{h_2}=0$. Hence this polynomial $f$ has the form
\[
f(x)=\sum_{0\leq i\leq k-1,i\neq h_2}f_ix^i+\eta_1f_{h_1}x^{k-1+t_1},
\]
and therefore $\deg\, f\le k-1+t_1$. By equation (\ref{eq:dr-common-zero}),
\[
d_2(\C_2)\ge n-(k-1+t_1)=n-k-t_1+1.
\]

Now let $3\leq r\leq k$, and define
\[
U={\rm span}_{\F_q}\{x^i:0\leq i\leq k-1,\ i\notin\{h_1,h_2\}\}.
\]
Then $\dim\, U=k-2$. For any $r$-dimensional subspace $W\leq\mathcal S_2$,
\[
\dim\,(W\cap U)\geq \dim\, W+\dim\, U-\dim\,(W+U)\geq  r+k-2-k=r-2.
\]
Let $I=I(W)$ be the common zero set of $W$, we have every polynomial in $W$ vanishes at all $\alpha_i$ with $i\in I$.
Hence every polynomial in $W\cap U$ is divisible by
\[
P_I(x)=\prod_{i\in I}(x-\alpha_i).
\]
Now define
\[
V_I=\left\{v(x)\in\F_q[x]\,:\,\deg v<k,\ P_I(x)\mid v(x)
\right\}.
\]
Similar to Theorem \ref{rangeofd(C)}, we have ${\rm dim}\,V_I=k-|I|$ and  $W\cap U\subseteq V_I$. Hence, 
Therefore,
\[
r-2\le\dim\,(W\cap U)\le k-|I|,
\]
and $|I|\le k-r+2$. By equation (\ref{eq:dr-common-zero}), it follows that
\[
d_r(\C)\ge n-k+r-2,\ 3\le r\le k.
\]

Finally, if $k\ge3$, then at least one monomial among
$1,x,\ldots,x^{k-1}$ is untwisted. As in Theorem \ref{rangeofd(C)}, one
checks that every coordinate occurs in the support of some codeword. Hence $\supp(\C)=[n]$. Since $\C$ itself is the unique $k$-dimensional subcode, we have $d_k(\C)=n$. This completes the proof.
\end{proof}

Next, we give a  classification of the generalized Hamming weights $d_r(C)$ for $r\ge 3$.
\begin{theorem}\label{thm:higher}
Let $\C_2$ be a TRS code defined in (\ref{def:C(a,1)2}), and
assume that $k\ge3$. Then, for $3\le r\le k$,
\[
d_r(\C_2)=
\begin{cases}
n-k+r-2, & 3\le r<k+3-d_2(\C_2^\perp),\\[1mm]
n-k+r-1, & k+3-d_2(\C^\perp)\le r<k+2-d_1(\C_2^\perp),\\[1mm]
n-k+r, & k+2-d_1(\C_2^\perp)\le r\le k.
\end{cases}
\]
\end{theorem}
\begin{proof}
By Theorem~\ref{thm:range}, for $3\le r\le k-1$, we have
\[
d_r(\C_2)\in\{n-k+r-2,\ n-k+r-1,\ n-k+r\},
\]
and $d_k(\C_2)=n$. Let
\[
r_1=\min\left\{r\,:\,3\le r\le k,\ d_r(\C_2)\ge n-k+r-1\right\},
\]
and
\[
r_2=\min\left\{r\,:\,3\le r\le k,\ d_r(\C_2)=n-k+r\right\}.
\]
By Lemma \ref{d1<d2<dots<dr},  if $d_r(\C_2)\ge n-k+r-1$, then
\[
d_{r+1}(\C_2)\ge d_r(\C_2)+1\ge n-k+r,
\]
and hence $d_{r+1}(\C_2)\ne n-k+(r+1)-2$. Similarly, once $d_r(\C_2)=n-k+r$, we have
\[
d_{r+1}(\C_2)\ge n-k+r+1,
\]
and by Lemma \ref{rMDS}, we have $d_{r+1}(\C_2)=n-k+r+1$. Therefore, for $r\geq 3$, the  generalized Hamming weights necessarily have the form
\[
d_r(\C_2)=\begin{cases}
n-k+r-2, & 3\le r<r_1,\\
n-k+r-1, & r_1\le r<r_2,\\
n-k+r,   & r_2\le r\le k.
\end{cases}
\]

We now determine the values of $r_1$ and $r_2$. First, we  consider the value of  $r_2$. Since $r_2$ is the first index for which $d_{r_2}(\C_2)=n-k+r_2$, and the integer $n-k+r_2-1$ is not a generalized Hamming weight of $\C_2$. On the other hand,
\[
d_{r_2}(\C_2)=n-k+r_2,\ d_{r_2+1}(\C_2)=n-k+r_2+1,\dots,\ d_k(\C_2)=n.
\]
Hence every integer greater than $n-k+r_2-1$ is a generalized Hamming weight of $\C_2$. Thus $n-k+r_2-1$ is the largest integer which is not a generalized Hamming weight of $\C_2$. By Lemma \ref{dual GHW of C}, the largest integer which is not a generalized Hamming weight of $\C_2$ is $n+1-d_1(\C_2^\perp)$. Therefore,
\[
n-k+r_2-1=n+1-d_1(\C_2^\perp),
\]
which gives $r_2=k+2-d_1(\C_2^\perp)$. Similarly, we have  $n-k+r_1-2$ is the second largest integer which is not a generalized Hamming
weight of $\C_2$. Again by Lemma \ref{dual GHW of C}, the second largest integer which is not a generalized Hamming weight of $\C_2$ is $n+1-d_2(\C_2^\perp)$. Therefore,
\[
n-k+r_1-2=n+1-d_2(\C_2^\perp),
\]
and hence $r_1=k+3-d_2(\C_2^\perp)$. This completes the proof.
\end{proof}

\begin{example}
Let $\C_2$ be a TRS code defined in (\ref{def:C(a,1)2}). Let $q=17$, $n=13$, $k=9$, $t_1=2$, $t_2=4$, $h_1=2$, $h_2=4$, and $\eta_1=\eta_2=-1$. Take $\bm{\alpha}=(1,2,4,8,9,13,15,16,0,3,5,6,7)$. With the help of SageMath, we have $d_1(\C_2^\perp)=6$ and $d_2(\C_2^\perp)=8$. Moreover, its generalized Hamming weight can be derived as follows:
\[
\bigl(d_1(\C_2),d_2(\C_2),d_3(\C_2),d_4(\C_2),d_5(\C_2),d_6(\C_2),d_7(\C_2),d_8(\C_2),d_9(\C_2)\bigr)=(3,4,5,7,9,10,11,12,13).
\]
This result is consistent with Theorem \ref{thm:higher}.
\end{example}

We next give explicit  conditions for determining the first two generalized Hamming weights of $\C_2$ and $\C_2^{\perp}$.

\begin{proposition}\label{prop:firsttwo d1 d2}
Let $\C_2$ be a TRS code defined in (\ref{def:C(a,1)2}). For any $I\subseteq[n]$ with $|I|=k-r+\mu$, where $0\leq \mu \leq t_2$ if $r=1$, and $0\leq \mu \leq t_1+1$ if $r=2$. Let 
\[
P_I(x)=\prod_{i\in I}(x-\alpha_i)=\sum_{j=0}^{k-r+\mu}c_jx^{k-r+\mu-j},
\]
and let $A_{r,\mu}(I)$ be the matrix
defined as follows:
\begin{equation*}
\left(
\begin{array}{cccccc}
c_1&c_0&\cdots&0\\
c_2&c_1&\cdots&0\\
\vdots&\vdots&&\vdots\\
c_{t_2-t_1-1}&c_{t_2-t_1-2}&\cdots&c_{\mu-r-t_1}\\
c_{t_2-t_1}-\eta_1c_{k-1+t_2-h_1}&c_{t_2-t_1-1}-\eta_1c_{k-2+t_2-h_1}&\cdots&c_{\mu-r-t_1+1}-\eta_1c_{k-r+\mu-h_1}\\
c_{t_2-t_1+1}&c_{t_2-t_1}&\cdots&c_{\mu-r-t_1+2}\\
\vdots&\vdots&&\vdots\\
c_{t_2-1}&c_{t_2-2}&\cdots&c_{\mu-r}\\
1-\eta_2c_{k-1+t_2-h_2}&-\eta_2c_{k-2+t_2-h_2}&\cdots&
-\eta_2c_{k-r+\mu-h_2}
\end{array}
\right).
\end{equation*}
Then $d_r(\C_2)=n-k+r-\mu$ if and only if the following conditions hold:
\begin{enumerate}
\item [(1)] There exists a $(k-r+\mu)$-subset $I_1\subseteq[n]$ such that ${\rm rank}\, A_{r,\mu}(I_1)\le t_2-\mu$.
\item [(2)] For every $(k-r+\mu+1)$-subset $I_2\subseteq[n]$, ${\rm rank}\, A_{r,\mu+1}(I_2)\ge t_2-\mu$.
\end{enumerate}
\end{proposition}

\begin{proof}
For an $r$-dimensional subcode $\mathcal{D}$ of $\C_2$, let
\[
N(\mathcal{D})=\{i\in[n]:c_i=0\text{ for every }c\in \mathcal{D}\}
\]
be its common zero set. By the definition of the generalized Hamming weight, we have
\[
d_r(\C_2)=n-\max_{\substack{\mathcal{D}\le\C_2\\ \dim\, \mathcal{D}=r}}|N(\mathcal{D})|.
\]
Hence $d_r(\C_2)=n-k+r-\mu$ if and only if $\max_{\mathcal{D}\le\C_2}|N(\mathcal{D})|=k-r+\mu$, i.e., there exist  an $r$-dimensional subcode $\mathcal{D}$ having at least $k-r+\mu$ common zero coordinates. Let  $I\subseteq[n]$ be a $(k-r+\mu)$-subset, and 
\[
P_I(x)=\prod_{i\in I}(x-\alpha_i)=\sum_{j=0}^{k-r+\mu}c_jx^{k-r+\mu-j}.
\]
An $r$-dimensional subcode $\mathcal{D}$ whose codewords vanish on all coordinates indexed by $I$ exists if and only if the space
\[
F_I=\{f\in\mathcal S_2:P_I(x)\mid f(x)\}
\]
has dimension at least $r$. Let $f(x)\in F_I$. Since $P_I(x)\mid f(x)$, there exists a polynomial $q(x)$ such that  $f(x)=P_I(x)q(x)$. Since every polynomial in $\mathcal S_2$ has degree at most $k-1+t_2$, we have
$\deg\, q(x)\le t_2+r-\mu-1$. We can assume that 
\[
q(x)=\sum_{j=0}^{t_2+r-\mu-1}b_jx^{t_2+r-\mu-1-j}.
\]
Then we have the coefficient of $x^{k-1+t_2-m}$ in $f(x)$ is
\begin{equation}\label{eq:coefficient of f}
f_{k-1+t_2-m}=\sum_{j=0}^{t_2+r-\mu-1}c_{m-j}b_j,
\end{equation}
where $0\leq m\leq k-1+t_2$. For any polynomial $f(x)\in \mathcal{S}_2$, we have 
\[
f_{k-1+u}=0\ \text{for } u\in\{1,\ldots,t_2-1\}\setminus\{t_1\},\ \text{and}\ f_{k-1+t_1}=\eta_1f_{h_1},\ f_{k-1+t_2}=\eta_2f_{h_2}.
\]
From equation (\ref{eq:coefficient of f}), the above equations can be simplified as 
\[
\sum_{j=0}^{t_2+r-\mu-1}c_{t_2-u-j}b_j=0,\ \sum_{j=0}^{t_2+r-\mu-1}\left(c_{t_2-t_1-j}-\eta_1c_{k-1+t_2-h_1-j}
\right)b_j=0.
\]
and 
\[
\left(1-\eta_2c_{k-1+t_2-h_2}\right)b_0-\eta_2c_{k-2+t_2-h_2}b_1-\cdots-\eta_2c_{k-r+\mu-h_2}b_{t_2+r-\mu-1}=0.
\]
Then for polynomial $f(x)\in \mathcal{S}_2$, we have 
\begin{equation}\label{eq:coefficient matrix}
\resizebox{0.95\textwidth}{!}{$
\left(
\begin{array}{cccccc}
c_1&c_0&\cdots&0\\
c_2&c_1&\cdots&0\\
\vdots&\vdots&&\vdots\\
c_{t_2-t_1-1}&c_{t_2-t_1-2}&\cdots&c_{\mu-r-t_1}\\
c_{t_2-t_1}-\eta_1c_{k-1+t_2-h_1}&c_{t_2-t_1-1}-\eta_1c_{k-2+t_2-h_1}&\cdots&c_{\mu-r-t_1+1}-\eta_1c_{k-r+\mu-h_1}\\
c_{t_2-t_1+1}&c_{t_2-t_1}&\cdots&c_{\mu-r-t_1+2}\\
\vdots&\vdots&&\vdots\\
c_{t_2-1}&c_{t_2-2}&\cdots&c_{\mu-r}\\
1-\eta_2c_{k-1+t_2-h_2}&-\eta_2c_{k-2+t_2-h_2}&\cdots&
-\eta_2c_{k-r+\mu-h_2}
\end{array}
\right)
\begin{pmatrix}
b_0\\
b_1\\
b_2\\
\vdots\\
b_{t_2+r-\mu-2}\\
b_{t_2+r-\mu-1}
\end{pmatrix}
=0.
$}
\end{equation}
The coefficient matrix in (\ref{eq:coefficient matrix}) is precisely $A_{r,\mu}(I)$. Hence
\[
\dim\, F_I=t_2+r-\mu-{\rm rank}\, A_{r,\mu}(I).
\]
Therefore, $\dim\, F_I\ge r$ if and only if
\[
t_2+r-\mu-{\rm rank}\, A_{r,\mu}(I)\ge r,
\]
that is, ${\rm rank}\, A_{r,\mu}(I)\le t_2-\mu$. Thus there exists an $r$-dimensional subcode of $\C_2$ having at least $k-r+\mu$ common zero coordinates if and only if there exists a $(k-r+\mu)$-subset $I_1\subseteq[n]$ such that ${\rm rank}\, A_{r,\mu}(I_1)\le t_2-\mu$. Moreover, there does not exist an $r$-dimensional subcode of $\C_2$ having at least $k-r+\mu+1$ common zero coordinates if and only if for any $(k-r+\mu+1)$-subset $I_2$, have
\[
{\rm rank}\, A_{r,\mu+1}(I_2)\ge t_2-\mu.
\]
This completes the proof.
\end{proof}

\begin{corollary}\label{cor:two-twist-special}
Let $\C_2$ be the TRS code defined in~\eqref{def:C(a,1)2} with $(t_1,t_2)=(1,2)$, $(h_1,h_2)=(k-2,k-1)$. For each $I\subseteq[n]$, define $P_I(x)=\prod\limits_{i\in I}(x-\alpha_i)=\sum\limits_{j=0}^{|I|}c_j(I)x^{|I|-j}$, and
\[
u(I)=c_1(I)-\eta_1c_3(I),\ v(I)=1-\eta_1c_2(I),\ w(I)=1-\eta_2c_2(I),\ z(I)=-\eta_2c_1(I),
\]
Furthermore, define
\[
M_1(I)=
\begin{pmatrix}
u(I)\\
w(I)
\end{pmatrix},\ M_2(I)=
\begin{pmatrix}
u(I) & v(I)\\
w(I) & z(I)
\end{pmatrix},\ M_3(I)
=\begin{pmatrix}
u(I) & v(I) & -\eta_1c_1(I)\\
w(I) & z(I) & -\eta_2
\end{pmatrix}.
\]
Then the following statements hold.
\begin{enumerate}
    \item[(1)] The minimum distance of $\C_2$ satisfies $d_1(\C_2)\in\{n-k-1,n-k,n-k+1\}$. More precisely, we have
    \begin{enumerate}
        \item[(i)] $d_1(\C_2)=n-k+1$ if and only if ${\rm rank}\, M_2(I)=2$ for every $k$-subset $I\subseteq[n]$.

        \item[(ii)] $d_1(\C_2)=n-k$ if and only if there exists a
        $k$-subset $I\subseteq[n]$ such that ${\rm rank}\, M_2(I)\le 1$,
        and ${\rm rank}\, M_1(J)=1$ for every $(k+1)$-subset $J\subseteq[n]$.

        \item[(iii)] $d_1(\C_2)=n-k-1$ if and only if there exists a
        $(k+1)$-subset $I\subseteq[n]$ such that ${\rm rank}\, M_1(I)=0$.
    \end{enumerate}

    \item[(2)] The second generalized Hamming weight of $\C_2$ satisfies $d_2(\C_2)\in\{n-k,n-k+1,n-k+2\}$. More precisely, we have
    \begin{enumerate}
        \item[(i)] $d_2(\C_2)=n-k+2$ if and only if ${\rm rank}\, M_3(I)=2$ for every $(k-1)$-subset $I\subseteq[n]$.

        \item[(ii)] $d_2(\C_2)=n-k+1$ if and only if there exists a
        $(k-1)$-subset $I\subseteq[n]$ such that ${\rm rank}\, M_3(I)\le 1$, and ${\rm rank}\, M_2(J)\ge 1$ for every $k$-subset $J\subseteq[n]$.

        \item[(iii)] $d_2(\C_2)=n-k$ if and only if there exists a
        $k$-subset $I\subseteq[n]$ such that ${\rm rank}\, M_2(I)=0$.
    \end{enumerate}
\end{enumerate}
\end{corollary}

\begin{example}
Let $\C_2$ be a TRS code defined in (\ref{def:C(a,1)2}). Let $q=13$, $n=10$, $k=6$, $t_1=1$, $t_2=3$, $h_1=0$,  $h_2=3$, and $\eta_1=4$, $\eta_2=7$. Take $\bm{\alpha}=(0,1,2,3,5,6,8,9,10,12)$. For
$r=1$ and $\mu_1=1$. In this case, take $I_1=\{1,2,4,5,6,9\}$. The matrix
$A_{1,1}(I_1)$ is
\[
A_{1,1}(I_1)
=
\begin{pmatrix}
1&1&0\\
6&1&1\\
9&9&0
\end{pmatrix}.
\]
Therefore, ${\rm rank}\,A_{1,1}(I_1)=2\le t_2-\mu_1$. Moreover, for  all subsets
$I_2\subseteq[n]$ satisfying
\[
|I_2|=k-r+\mu_1+1=6-1+1+1=7.
\]
With the help of SageMath, we have ${\rm rank}\, A_{1,2}(I_2)\ge 2$ and  $d_1(\C_2)=4$. For $r=2$ and $\mu_2=0$. In this case, take $I_1=\{1,2,3,4\}$. The matrix
$A_{2,0}(I_1)$ is
\[
A_{2,0}(I_1)
=
\begin{pmatrix}
7&1&0&0&0\\
11&7&1&0&0\\
1&0&3&1&3
\end{pmatrix}.
\]
Thus ${\rm rank}\, A_{2,0}(I_1)=3\le t_2-\mu_2$. Moreover, for all subsets $I_2\subseteq[n]$ satisfying
\[
|I_2|=k-r+\mu_2+1=6-2+0+1=5.
\]
With the help of SageMath, we have ${\rm rank}\, A_{2,1}(I_2)\ge3$ and $d_2(\C_2)=6$. These results are consistent with Proposition \ref{prop:firsttwo d1 d2}.
\end{example}

\begin{proposition}\label{prop:dualweights d1d2}
Let $\C_2$ be a TRS code defined in (\ref{def:C(a,1)2}), and let $A_{j,p-k+j}(I)$ be the coefficient matrix defined above for a $p$-subset $I\subseteq[n]$. For $j\in\{1,2\}$ and $j\le \nu \le k+j$, we have $d_j(\C_2^\perp)=\nu $ if and only if the following conditions hold:
\begin{enumerate}
\item[(1)] for every $(\nu-1)$-subset $J_1\subseteq[n]$, ${\rm rank}\, A_{j,\nu-k+j-1}(J_1)\ge t_2-j+1$;
\item[(2)] there exists an $\nu$-subset $J_2\subseteq[n]$ such that ${\rm rank}\, A_{j,\nu-k+j}(J_2)=t_2-j$.
\end{enumerate}
\end{proposition}

\begin{proof}
Let $G$ be a generator matrix of $\C_2$. Since $G$ is  a parity check matrix of $\C_2^\perp$, by Lemma \ref{rank-characterization}, $d_j(\C_2^\perp)=\nu$ if and only if
\begin{enumerate}
\item[(1)] any $\nu-1$ columns of $G$ have rank at least $\nu-j$;
\item[(2)] there exist $\nu$ columns of $G$ having rank $\nu-j$.
\end{enumerate}

For a subset $I\subseteq[n]$ with $|I|=p$, let
\[
F_I=\{f(x)\in \mathcal{S}_2\,:\,P_I(x)\mid f(x)\},\ P_I(x)=\prod_{i\in I}(x-\alpha_i).
\]
Hence, $F_I$ consists exactly of those polynomials in $\mathcal{S}_2$ whose evaluations vanish at all coordinates indexed by $I$. Let $g_0(x),g_1(x),\ldots,g_{k-1}(x)$ be the basis of $\mathcal{S}_2$. Then every polynomial $f(x)\in \mathcal{S}_2$ can be written uniquely as
\[
f(x)=a_0g_0(x)+a_1g_1(x)+\cdots+a_{k-1}g_{k-1}(x).
\]
Let $G_I$ denote the submatrix of $G$ consisting of the columns indexed by $I$. The evaluations of $f(x)$ at the points indexed by
$I$ are given by
\[
(f(\alpha_i))_{i\in I}=(a_0,a_1,\ldots,a_{k-1})G_I.
\]
Therefore, $f(x)\in F_I$ if and only if $(a_0,a_1,\ldots,a_{k-1})G_I=0$. Hence $F_I$ is naturally identified with the solution space of this
homogeneous linear system. Therefore,
\[
\dim\, F_I=k-{\rm rank}\, G_I.
\]
On the other hand, in Proposition \ref{prop:firsttwo d1 d2}, we have proven that 
\[
\dim\, F_I=k+t_2-p-{\rm rank}\, A_{j,p-k+j}(I).
\]
Hence, we have
\[
{\rm rank}\, G_I=p-t_2+{\rm rank}\, A_{j,p-k+j}(I).
\]
Therefore, for an $(\nu-1)$-subset $J_1\subseteq[n]$, ${\rm rank}\, G_{J_1}\ge \nu-j$ if and only if ${\rm rank}\, A_{j,\nu-k+j-1}(J_1)\ge t_2-j+1$. Similarly, for an $\nu$-subset $J_2\subseteq[n]$, ${\rm rank}\, G_{J_2}=\nu-j$ if and only if ${\rm rank}\, A_{j,\nu-k+j}(J_2)=t_2-j$.  This completes the proof.
\end{proof}

\begin{corollary}
Let $\C_2$ be the TRS code defined in~\eqref{def:C(a,1)2} with $(t_1,t_2)=(1,2)$, $(h_1,h_2)=(k-2,k-1)$, where $\eta_1,\eta_2\in\mathbb{F}_q^*$.
For each $I\subseteq[n]$, let $P_I(x)=
\prod\limits_{i\in I}(x-\alpha_i)=
\sum\limits_{j=0}^{|I|}c_j(I)x^{|I|-j}$, and 
\[
u(I)=c_1(I)-\eta_1c_3(I),\ v(I)=1-\eta_1c_2(I), \ w(I)=1-\eta_2c_2(I),
\ z(I)=-\eta_2c_1(I).
\]
Set
\[
M_1(I)=
\begin{pmatrix}
u(I)\\
w(I)
\end{pmatrix},\ M_2(I)=
\begin{pmatrix}
u(I) & v(I)\\
w(I) & z(I)
\end{pmatrix},\ M_3(I)=
\begin{pmatrix}
u(I) & v(I) & -\eta_1c_1(I)\\
w(I) & z(I) & -\eta_2
\end{pmatrix}.
\]
Then the following statements hold.

\begin{enumerate}
    \item[(1)] The minimum distance of $\C_2^\perp$ satisfies $d_1(\C_2^\perp)\in\{k-1,k,k+1\}$. More precisely, we have
    \begin{enumerate}
        \item[(i)] $d_1(\C_2^\perp)=k-1$ if and only if there exists a
        $(k-1)$-subset $I\subseteq[n]$ such that ${\rm rank}\, M_3(I)=1$.

        \item[(ii)] $d_1(\C_2^\perp)=k$ if and only if ${\rm rank}\, M_3(I)=2$ for every $(k-1)$-subset $I\subseteq[n]$, and there exists an$k$-subset $J\subseteq[n]$ such that ${\rm rank}\, M_2(J)=1$.

        \item[(iii)] $d_1(\C_2^\perp)=k+1$ if and only if ${\rm rank}\, M_2(I)=2$  for every $k$-subset $I\subseteq[n]$, and there exists a $(k+1)$-subset $J\subseteq[n]$ such that ${\rm rank}\, M_1(J)=1$.
    \end{enumerate}

    \item[(2)] The second generalized Hamming weight of $\C_2^\perp$ satisfies $d_2(\C_2^\perp)\in\{k,k+1,k+2\}$. More precisely, we have
    \begin{enumerate}
        \item[(i)] $d_2(\C_2^\perp)=k$ if and only if there exists a
        $k$-subset $I\subseteq[n]$ such that  $M_2(I)=0$.
    
        \item[(ii)] $d_2(\C_2^\perp)=k+1$ if and only if $M_2(I)\neq0$ for every $k$-subset $I\subseteq[n]$, and  there exists a $(k+1)$-subset $J\subseteq[n]$ such that
        $M_1(J)=0$.

        \item[(iii)] $d_2(\C_2^\perp)=k+2$ if and only if $M_1(I)\neq0$
        for every $(k+1)$-subset $I\subseteq[n]$.
    \end{enumerate}
\end{enumerate}
\end{corollary}

\begin{example}
Let $\C_2$ be a TRS code defined in (\ref{def:C(a,1)2}). Let $q=11$, $n=9$, $k=5$, $t_1=2$, $t_2=4$, $h_1=1$, $h_2=3$, and $\eta_1=3$, $\eta_2=8$. Take $\bm{\alpha}=(0,1,2,4,5,6,7,9,10)$. With the help of SageMath, we have $d_1(\C_2^{\perp})=5$ and $d_2(\C_2^{\perp})=7$. Take $I_1=\{1,2,4,5,9\}$, and $I_2=\{1,2,3,4,5,6,7\}$, we have
\[
A_{1,1}(I_1)=\begin{pmatrix}
2&1&0&0\\
8&2&1&5\\
9&8&2&1\\
1&6&5&2
\end{pmatrix},\
A_{2,4}(I_2)=\begin{pmatrix}
8&1\\
5&6\\
2&5\\
9&8
\end{pmatrix},
\]
and ${\rm rank}\, A_{1,1}(I_1)=3=t_2-1$ and ${\rm rank}\, A_{2,4}(I_2)=2=t_2-2$. Moreover, for any $4$-subset $J_1$ and $6$-subset $J_2$, have ${\rm rank}\, A_{1,0}(J_1)\ge t_2-1+1=4$ and ${\rm rank}\, A_{2,3}(J_2)\ge t_2-2+1=3$.
\end{example}

With the minimum distances of $\mathcal{C}_2$ and its dual characterized above, we have a complete description of the generalized Hamming weight hierarchy of $\mathcal{C}_2$.

\begin{theorem}\label{thm:complete}
Let $\C_2$ be a TRS code defined in (\ref{def:C(a,1)2}). If $\C_2$ is a MDS code, then
\[
d_r(\C_2)=n-k+r,\qquad 1\le r\le k.
\]
Suppose that $\C_2$ is non-MDS. Then the generalized Hamming weight hierarchy of $\C_2$ is given by
\[
d_r(\C_2)=
\begin{cases}
n-k+1-\mu_1,
& r=1,\\[1mm]
n-k+2-\mu_2,
& r=2,\\[1mm]
n-k+r-2,
& 3\le r<k+3-\nu_2,\\[1mm]
n-k+r-1,
& k+3-\nu_2\le r<k+2-\nu_1,\\[1mm]
n-k+r,
& k+2-\nu_1\le r\le k,
\end{cases}
\]
where $\mu_1$, $\mu_2$ and $\nu_1$, $\nu_2$ are determined by Proposition \ref{prop:firsttwo d1 d2} and \ref{prop:dualweights d1d2}.
\end{theorem}

\begin{corollary}
With notation in Proposition \ref{prop:dualweights d1d2}. Then we have
\[
\nu_1\ge \max\{2,h_1+1,k-h_2\},\ \nu_2\ge \max\{3,h_1+2,k-h_2+1\}.
\]
Consequently, $d_r(\C_2)=n-k+r$ for every $3\le r\le k$ satisfying $r\ge k+2-\max\{2,h_1+1,k-h_2\}$.
\end{corollary}

\begin{proof}
Let $G$ be a generator matrix of $\C_2$. For any $p$-subset $I\subseteq [n]$, if  $p\le h_1$, the rows of $G_I$ corresponding to $1,x,\ldots,x^{p-1}$ form a nonsingular Vandermonde matrix. Hence ${\rm rank}\, G_I=p$, and therefore $\nu_1\ge h_1+1$.

On the other hand, if $p\le k-h_2-1$, then by Proposition \ref{prop:dualweights d1d2}, we have ${\rm rank}\, A_{1,p-k+1}(I)=t_2$. Hence,
\[
{\rm rank}\, G_I=p-t_2+{\rm rank}\, A_{1,p-k+1}(I)=p.
\]
Thus $\nu_1\ge k-h_2$. Moreover, $G$ has no zero column, so $\nu_1\ge2$. Therefore,
\[
\nu_1\ge\max\{2,h_1+1,k-h_2\}.
\]
Since the generalized Hamming weights are strictly increasing, then $\nu_2\ge \nu_1+1$.

Finally, Theorem~\ref{thm:complete} gives $d_r(\C_2)=n-k+r$ whenever $r\ge k+2-\nu_1$. Then $d_r(\C_2)=n-k+r$ for every $3\le r\le k$ satisfying $r\ge k+2-\max\{2,h_1+1,k-h_2\}$.
\end{proof}

\begin{table*}[htbp]
\centering
\caption{Generalized Hamming weight hierarchies of two-twist TRS codes}
\label{tab:two-twist-examples}
\resizebox{\textwidth}{!}{
\begin{tabular}{ccccccccccccccc}
\toprule
$q$ & $n$ & $k$
& $t_1$ & $t_2$
& $h_1$ & $h_2$
& $\eta_1$ & $\eta_2$
& $\mu_1$ & $\mu_2$
& $\nu_1$ & $\nu_2$
& $\boldsymbol{\alpha}$
& $(d_1,d_2,\dots,d_k)$
\\
\midrule
7 & 5 & 3 & 1 & 2 & 1 & 2 & 1 & 2
& 1 & 0 & 3 & 5
& $(0,1,2,4,5)$
& $(2,4,5)$ \\

7 & 5 & 3 & 1 & 2 & 1 & 2 & 1 & 2
& 1 & 1 & 2 & 5
& $(0,1,2,4,6)$
& $(2,3,5)$ \\

7 & 5 & 3 & 1 & 2 & 1 & 2 & 1 & 2
& 2 & 1 & 2 & 4
& $(0,1,2,3,4)$
& $(1,3,5)$ \\

7 & 6 & 4 & 1 & 2 & 1 & 2 & 3 & 2
& 1 & 0 & 4 & 6
& $(0,1,2,3,4,5)$
& $(2,4,5,6)$ \\

7 & 7 & 4 & 1 & 3 & 1 & 2 & 3 & 2
& 1 & 0 & 4 & 6
& $(0,1,2,3,4,5,6)$
& $(3,5,6,7)$ \\

11 & 8 & 4 & 2 & 4 & 1 & 2 & 1 & 5
& 1 & 0 & 4 & 6
& $(0,1,2,3,4,5,6,7)$
& $(4,6,7,8)$ \\

11 & 8 & 4 & 2 & 4 & 1 & 2 & 1 & 5
& 2 & 0 & 4 & 5
& $(0,1,2,3,4,5,6,8)$
& $(3,6,7,8)$ \\

$11$ & $9$ & $5$
& $1$ & $2$
& $1$ & $3$
& $3$ & $8$
& $1$ & $0$
& $5$ & $7$
& $(0,1,2,3,4,5,6,8,9)$
& $(4,6,7,8,9)$
\\

$11$ & $9$ & $5$
& $1$ & $2$
& $1$ & $3$
& $3$ & $8$
& $1$ & $1$
& $4$ & $7$
& $(0,1,2,3,4,6,7,9,10)$
& $(4,5,7,8,9)$
\\

$11$ & $9$ & $5$
& $1$ & $2$
& $1$ & $3$
& $3$ & $8$
& $2$ & $0$
& $5$ & $6$
& $(0,1,2,3,4,6,7,8,10)$
& $(3,6,7,8,9)$
\\

$11$ & $9$ & $5$
& $1$ & $2$
& $1$ & $3$
& $3$ & $8$
& $2$ & $1$
& $4$ & $6$
& $(0,1,2,3,4,5,6,7,8)$
& $(3,5,7,8,9)$
\\
$13$ & $10$ & $5$
& $1$ & $3$
& $1$ & $2$
& $4$ & $7$
& $1$ & $0$
& $5$ & $7$
& $(0,1,2,3,4,5,6,7,8,12)$
& $(5,7,8,9,10)$
\\

$13$ & $10$ & $5$
& $1$ & $3$
& $1$ & $2$
& $4$ & $7$
& $1$ & $1$
& $4$ & $7$
& $(0,1,2,3,4,5,6,9,10,11)$
& $(5,6,8,9,10)$
\\

$13$ & $10$ & $5$
& $1$ & $3$
& $1$ & $2$
& $4$ & $7$
& $2$ & $0$
& $5$ & $6$
& $(0,1,2,3,4,5,6,7,8,10)$
& $(4,7,8,9,10)$
\\

$13$ & $10$ & $5$
& $1$ & $3$
& $1$ & $2$
& $4$ & $7$
& $2$ & $1$
& $4$ & $6$
& $(0,1,2,3,4,5,6,7,8,9)$
& $(4,6,8,9,10)$
\\

13 & 11 & 6 & 2 & 4 & 0 & 2 & 3 & 4
& 1 & 0 & 6 & 8
& $(0,1,2,3,4,5,6,8,9,10,11)$
& $(5,7,8,9,10,11)$ \\

13 & 11 & 6 & 2 & 4 & 0 & 2 & 3 & 4
& 1 & 1 & 5 & 8
& $(0,1,2,3,4,5,6,7,8,9,10)$
& $(5,6,8,9,10,11)$ \\

13 & 11 & 6 & 2 & 4 & 0 & 2 & 3 & 4
& 2 & 1 & 5 & 7
& $(0,1,2,3,4,5,6,7,8,9,11)$
& $(4,6,8,9,10,11)$ \\

17 & 15 & 6 & 3 & 5 & 1 & 5 & 1 & 2
& 2 & 1 & 5 & 7
& $(0,1,2,3,4,5,6,7,8,9,10,11,12,13,14)$
& $(8,10,12,13,14,15)$ \\
\bottomrule
\end{tabular}
}
\end{table*}

\section{Subcode Support Weight Distribution of TRS Codes}
In this section, we determine the subcode support weight distributions of the  TRS code $\mathcal{C}_1$. More precisely, for each $1\le r\le k$, we determine the number of $r$-dimensional subcodes of $\mathcal{C}_1$ having
a prescribed support size. Since $\mathcal{C}_1$ is obtained by evaluating the polynomial space $\mathcal{S}_1$ defined in (\ref{S}) at the evaluation point set $\{\alpha_1,\alpha_2,\ldots,\alpha_n\}$, this problem can be reduced to counting $r$-dimensional subspaces of $\mathcal{S}_1$ whose polynomials vanish  on prescribed subsets of the evaluation point set.

For this purpose, for each subset $I\subseteq[n]$, we consider the subspace of $\mathcal{S}_1$ consisting of all polynomials divisible by $P_I(x)=\prod\limits_{i\in I}(x-\alpha_i)$. The following lemma expresses the subcode support weight distribution in terms of the dimensions of these subspaces.
\begin{lemma}\label{sswd of C}
	Let $\mathcal{C}_1$ be a TRS code defined in (\ref{def:C(a,1)}). For $1\leq r\leq k$, define 
	\[
	A_i^r=|\{\mathcal{D}\,:\,\mathcal{D}\in [\mathcal{C}_1,r],\ |{\rm supp}(\mathcal{D})|=i\}|.
	\]
	For any subset $I\subseteq [n]$, let $P_I(x)=\prod\limits_{i\in I}(x-\alpha_i)$ and $F_I=\{f(x)\in \mathcal{S}_1\,:\, P_I(x)|f(x)\}$. Then for $0\leq s\leq n$, we have
	\[
	A_{n-s}^r=\sum_{j=s}^n(-1)^{j-s}\binom{j}{s} \sum_{|I|=j}\left[\begin{matrix} {\rm dim}\,F_I  \\ r 
	 \end{matrix}\right]_q
	\]
\end{lemma}
\begin{proof}
	For an $r$-dimensional subcode $\mathcal{D}$, define
	\[N(\mathcal{D})= \{i\,:\, 1\leq i\leq n,\ c_i=0\ {\rm for}\ \forall \ {\bm d}=(d_1,d_2,\dots,d_n)\in \mathcal{D} \}.\]
	Then we have ${\rm supp}(\mathcal{D})=[n]\setminus N(\mathcal{D})$. Therefore, an $r$-dimensional subcode has support weight $n-s$ if and only if $|N(\mathcal{D})|=s$, where $0\leq s\leq n$. Hence, we have
	\[
	A_{n-s}^r=|\{\mathcal{D}\,:\,\mathcal{D}\in [\mathcal{C}_1,r],\ |N(\mathcal{D})|=s\}|.
	\]
	
	Next, we consider the size of set $A_{n-s}^r$. For a fixed set $I\subseteq [n]$, define
	\[
	S_r(I)=|\{\mathcal{D}\,:\,\mathcal{D}\in [\mathcal{C}_1,r],\ I\subseteq N(\mathcal{D})\}|.
	\]
	Since every subcode $\mathcal{D}$ corresponds a subspace $W$ of $\mathcal{S}_1$, and $I\subseteq N(\mathcal{D})$ means for any $f\in W$, $i\in I$, have $f(\alpha_i)=0$. Therefore, we have $P_I(x)\mid f(x)$, i.e., $I\subseteq N(\mathcal{D})$ if and only if $W\subseteq F_I$. It follows that $S_r(I)$ is exactly the number of $r$-dimensional subspaces of the $F_I$, then we have 
	\begin{equation}\label{eq:number of rdim subspace}
		S_r(I)=\left[\begin{matrix} {\rm dim}\,F_I  \\ r 
	 \end{matrix}\right]_q.
	\end{equation}
	For a fixed set $J\subseteq [n]$, define
	\[ E_r(J)=|\{\mathcal{D}'\,:\,\mathcal{D}'\in [\mathcal{C}_1,r],\ J= N(\mathcal{D}')\}|.\]
	Then  every subcode $\mathcal{D}'$ with $N(\mathcal{D}')=J\supseteq I$ is counted in $S_r(I)$. Hence, we have
	\[
	S_r(I)=\sum_{I\subseteq J}E_r(J).
	\]
	By the inclusion-exclusion principle, we have 
	\begin{equation}\label{eq:Er}
		E_r(I)=\sum_{I\subseteq J}(-1)^{|J|-|I|}S_r(J).
	\end{equation}
	It is known that $A_{n-s}^r$ denote the number of $r$-dimensional subcodes of $\mathcal{C}_1$ with common zero set has cardinality exactly $s$, from equation (\ref{eq:number of rdim subspace}) and (\ref{eq:Er}), we have
	\begin{equation}\label{eq:An-s}
		 A_{n-s}^r=\sum_{|I|=s}E_r(I)=\sum_{|I|=s}\sum_{I\subseteq J}(-1)^{|J|-|I|}\left[\begin{matrix} {\rm dim}\,F_J  \\ r 
	 \end{matrix}\right]_q
	\end{equation}
   Since a fixed $j$-subset $J$ contains $\binom{j}{s}$ $s$-subset, equation (\ref{eq:An-s}) can be simplified as
   \[
   A_{n-s}^r=\sum_{j=s}^n(-1)^{j-s}\binom{j}{s}\sum_{|J|=j}\left[\begin{matrix} {\rm dim}\,F_J  \\ r 
	 \end{matrix}\right]_q,
   \]
   This completes the proof.
\end{proof}

It follows from Lemma \ref{sswd of C} that the subcode support weight distribution is completely determined once the dimensions of the spaces
$F_I$ are known for all subsets $I\subseteq[n]$. We now
determine the dimension of  $F_I$ explicitly.

\begin{proposition}\label{dimFI}
Let $\mathcal{C}_1$ be a TRS code defined in (\ref{def:C(a,1)}). Let $I\subseteq[n]$ be a $p$-subset with $0\leq p\leq k-1+t$, and $P_I(x)=\prod_{i\in I}(x-\alpha_i)=\sum_{j=0}^{p}c_jx^{p-j}$. Let
\[
F_I
=
\{f(x)\in\mathcal S\,:\,P_I(x)\mid f(x)\}.
\]
For $m\geq1$, define
\[
q_m=\sum_{\ell=1}^{m}(-1)^\ell\sum_{\substack{i_1+\cdots+i_\ell=m}}
c_{i_1}c_{i_2}\cdots c_{i_\ell},\ q_0=1.
\]
Then
\[
\dim\, F_I
=
\begin{cases}
d_p+1,
&
\begin{array}{l}
\text{if }
c_{p-h}=c_{p-h+1}=\cdots=c_{k-h-1}=0,\\[1mm]
\displaystyle
\eta\sum_{j=0}^{\min\{p-k+h,t-1\}}
c_{k-h+j}q_{t-1-j}=1,\\[3mm]
\displaystyle
\eta\sum_{m=0}^{k-1+t-p}
q_m c_{k-1+t-b-m}=0,
\quad k\leq b\leq p-1,
\end{array}
\\[7mm]
d_p,
&
\text{otherwise},
\end{cases}
\]
where $d_p=\max\{k-p,0\}$.
\end{proposition}

\begin{proof}
Let 
\[
U=\operatorname{span}_{\mathbb F_q}\{x^i:0\leq i\leq k-1,\ i\neq h\},\ g(x)=x^h+\eta x^{k-1+t}.
\]
Then $\mathcal S=U\oplus\langle g(x)\rangle$. It is easy to verify that 
\begin{equation}\label{eq:dimofFI}
	{\rm dim}\,F_I={\rm dim}\,(F_I\cap U)+{\rm dim}\, (F_I/ (F_I\cap U)),\ {\rm and }\ {\rm dim}\, (F_I/ (F_I\cap U))\in \{0,1\}.
\end{equation}
Let
\[
V_I=\{v(x)\in\mathbb F_q[x]_{<k}\,:\,P_I(x)\mid v(x)\}.
\]
If $p<k$, We have 
\[{\rm dim}\,V_I={\rm dim}\,{\mathbb{F}_q}[x]_{<k-p}=k-p.\]
If $p\geq k$, then $V_I=\{0\}$. Therefore,
\[
\dim V_I=d_p=\max\{k-p,0\}.
\]
We first determine the dimension of $F_I\cap U$. $F_I\cap U$ denotes the set of polynomial which is divisible by $P_I(x)$ in $U$, then $F_I\cap U=V_I\cap U$. Then we have the following cases on the dimension of ${\rm dim}\,F_I$:

\noindent\textbf{Case 1.} If $V_I\subseteq U$. In this case,
\[
\dim\,(F_I\cap U)=\dim\,(V_I\cap U)=\dim\,V_I=d_p.
\]
For $p<k$ and any $v(x)\in V_I$, we have
\[v(x)=q(x)P_I(x)=\sum_{\ell=0}^{k-p-1}b_\ell x^{k-p-1-\ell}\sum_{j=0}^{p}c_jx^{p-j}=\sum_{s=0}^{k-1}\sum_{\ell=0}^{k-1-p}b_{\ell}c_{s-\ell}x^{k-1-s}.\]
If $v(x)\in U$, then 
\begin{equation}\label{eq:c}
v_h=\sum_{\ell=0}^{k-p-1}b_\ell c_{k-1-h-\ell}=0,
\end{equation}
where $v_h$ denotes the coefficient of $x^h$ in polynomial $v(x)$.  Hence $V_I\subseteq U$ if and only if
the dimension of system equation (\ref{eq:c}) is $d_p=k-p$, i.e., 
\begin{equation}\label{eq:coefficients c}
	c_{p-h}=c_{p-h+1}=\cdots=c_{k-h-1}=0.
\end{equation}
For $p\geq k$, $V_I=\{0\}$, so this condition is vacuous.

We next determine the dimension of $F_I/(F_I\cap U)$, the space $F_I$ contains an element outside $U$ if and only if there exists a $u(x)\in U$ such that
\begin{equation}\label{eq:divisible}
	P_I(x)\mid g(x)+u(x).
\end{equation}
Indeed, if $f(x)\in F_I\setminus U$, then $f(x)=u(x)+a g(x)$ for some $u(x)\in U$ and $a\in\mathbb F_q^*$. Multiplying by $a^{-1}$, we may assume that $g(x)+u(x)\in F_I$ for some $u(x)\in U$. Divide $g(x)$ directly by $P_I(x)$ and write
\[
g(x)=P_I(x)Q(x)+r(x),\ \deg\, r(x)<p, 
\]
then $P_I(x)\mid g(x)-r(x)$. From (\ref{eq:divisible}), this is equivalent to $P_I(x)\mid u(x)+r(x)$.

If $p<k$, then $\deg\,(u(x)+r(x))<k$, and hence
$u(x)+r(x)\in V_I$. Since $V_I\subseteq U$ and $u(x)\in U$, it follows that $r(x)\in U$.

If $p\geq k$, then
\[
\deg(u(x)+r(x))<p=\deg P_I(x).
\]
Therefore, we have $u(x)+r(x)=0$, and hence again $r(x)\in U$.

Conversely, if $r(x)\in U$, then by taking $u(x)=-r(x)\in U$, we obtain
\[
g(x)+u(x)=g(x)-r(x)=P_I(x)q(x)\in F_I.
\]
Moreover, $g(x)-r(x)\notin U$, since $r(x)\in U$ has degree less than $k$, whereas $g(x)$ contains the
nonzero term $\eta x^{k-1+t}$. Hence $g(x)-r(x)\in F_I\setminus U$.

Thus, $F_I$ contains an element outside $U$ if and only if $r(x)\in U$.   It suffices in the following to consider the equivalent condition for $r(x)\in U$. From ${\rm deg}\,g(x)=k-1+t$ and ${\rm deg}\,P_I(x)=p$, we have ${\rm deg}\,q(x)=k-1+t-p$, and we assume that $Q(x)=\eta\sum\limits_{m=0}^{k-1+t-p}q_mx^{k-1+t-p-m}$, where $q_0=1$, then 
\begin{equation}\label{eq:pq}
	P_I(x)Q(x)=\eta \sum_{a=0}^{p}c_ax^{p-a} \sum_{j=0}^{k-1+t-p}q_jx^{k-1+t-p-j}=\eta \sum_{a=0}^{p}\sum_{j=0}^{k-1+t-p}c_aq_jx^{k-1+t-a-j}.
\end{equation}
It is known that ${\rm deg}\, r(x)<p$, and comparing the coefficients of the term $x^{k-1+t-s}$, $1\leq s\leq k-1+t-p$ in equation (\ref{eq:pq}) and polynomial $g(x)$. Then we have
\begin{equation}\label{eq:cq}
	c_0q_s+c_1q_{s-1}+\cdots+c_{{\rm min}\{s,p\}}q_{s-{\rm min}\{s,p\}}=0,
\end{equation}
and equation (\ref{eq:cq}) has the solution
\[
q_s=\sum_{\ell=0}^s(-1)^{\ell}\sum_{i_1+\cdots+i_{\ell}=s}c_{i_1}c_{i_2}\cdots c_{i_{\ell}}.
\]
Therefore, 
\begin{equation}\label{eq:r(x)}
	r(x)=g(x)-P_I(x)q(x)=\sum_{b=0}^{p-1}\left(\delta_{b,h}-\eta \sum_{j=0}^{k-1+t-p}q_jc_{k-1+t-b-j}\right)x^b,
\end{equation}
where $\delta_{b,h}=1$ if $b=h$ and $\delta_{b,h}=0$ if $b\neq h$. Let $r_h$ denote the coefficient of term $x^h$  in $r(x)$, then $r(x)\in U$ if and only if 
\begin{equation}\label{eq:r(x) in U}
	\eta \sum_{j=0}^{k-1+t-p}q_jc_{k-1+t-h-j}=\eta \sum_{j=0}^{k-1+t-p}c_{k-1+t-h-j} \sum_{\ell=1}^j(-1)^{\ell}\sum_{i_1+\cdots+i_{\ell}=j}c_{i_1}c_{i_2}\cdots c_{i_{\ell}}=1
\end{equation}
Since $c_{k-1+t-h-j}=0$ for $j=0,1,\dots,k-1-p$, the equation (\ref{eq:r(x) in U}) can be simplified as 
\begin{equation}\label{eq:simplify r(x) in U}
	\eta \sum_{j=0}^{m}c_{k-h+j} \sum_{\ell=1}^{t-1-j}(-1)^{\ell}\sum_{i_1+\cdots+i_{\ell}=t-1-j}c_{i_1}c_{i_2}\cdots c_{i_{\ell}}=1,
\end{equation}
where $m={\rm min}\,\{p-k+h,t-1\}$. 

Moreover, if $p\geq k$, then $r(x)$ may contain terms of degrees
$k,k+1,\ldots,p-1$. Since every polynomial in $U$ has degree less than
$k$, all these coefficients must vanish. Thus, for every
\[
k\leq b\leq p-1,
\]
we require
\begin{equation}\label{eq:p geq k}
	\eta\sum_{m=0}^{k-1+t-p}q_m c_{k-1+t-b-m}=0.
\end{equation}
When $p<k$, there are no such $b$, so equation (\ref{eq:p geq k}) is vacuous.

Therefore,
\[
\dim\, F_I={\rm dim}\,(V_I\cap U)+{\rm dim}\,(F_I/ (F_I\cap U) )=d_p+1
\]
whenever (\ref{eq:coefficients c}), (\ref{eq:simplify r(x) in U})  and (\ref{eq:p geq k})  hold, we have $\dim\, F_I=d_p$ otherwise.

\noindent\textbf{Case 2.} If $V_I\not\subseteq U$. Since $V_I=\{0\}$ for $p\ge k$, this case can occur only when $p<k$. Hence, we have $d_p=k-p$.  For $p<k$, every polynomial $v(x)\in V_I$ can be uniquely written as
$v(x)=P_I(x)q(x)$,
where
\[
q(x)=\sum_{\ell=0}^{k-p-1}
b_\ell x^{k-p-1-\ell}.
\]
Then the coefficient $v_h$ of $x^h$ in $v(x)$ is
\[
v_h=\sum_{\ell=0}^{k-p-1}b_\ell c_{k-1-h-\ell}.
\]
Hence $v(x)\in V_I\cap U$ if and only if
\begin{equation}\label{eq:solution}
	\sum_{\ell=0}^{k-p-1}
b_\ell c_{k-1-h-\ell}=0.
\end{equation}
Since $V_I\not\subseteq U$, the coefficients $c_{k-1-h},c_{k-2-h},\ldots,c_{p-h}$ are not all zero. Therefore, the above equation (\ref{eq:solution}) is a nontrivial homogeneous linear equation in the $k-p$ variables $b_0,b_1,\ldots,b_{k-p-1}$. Its solution space has dimension $k-p-1$. Then we obtain
\[
\dim(V_I\cap U)=k-p-1.
\]
It remains to show that the space $F_I$ contains an element outside $U$. Since  the coefficients $c_{k-1-h},c_{k-2-h},\ldots,c_{p-h}$ are not all zero. Hence there exists some
$\ell_0\in\{0,1,\ldots,k-p-1\}$ such that $c_{k-1-h-\ell_0}\neq 0$. Let $g(x)=P_I(x)Q(x)+r(x)$,
and 
\[
b_{\ell_0}=\frac{r_h}
{c_{k-1-h-\ell_0}},\  b_\ell=0\ \text{for } \ell\neq \ell_0,
\]
where $r_h$ denotes the coefficient of $x^h$ in polynomial $r(x)$. Then for  $v(x)\in V_I$, we have 
\[
v_h=\sum_{\ell=0}^{k-p-1}b_\ell c_{k-1-h-\ell}=r_h.
\]
Now set $u(x)=v(x)-r(x)$. Since $p<k$, we have
$\deg\, v(x)<k\ \text{and}\ \deg\, r(x)<p<k,$
and hence $\deg \, u(x)<k$.
Moreover,
\[u_h=
v_h-r_h=0,
\]
Thus $u(x)\in U.$ Then we can  obtain
\[
\begin{aligned}
g(x)+u(x)
&=P_I(x)Q(x)+r(x)+v(x)-r(x)\\
&=P_I(x)Q(x)+v(x).
\end{aligned}
\]
Since $v(x)\in V_I$, we have $P_I(x)\mid v(x)$. Therefore,
$P_I(x)\mid g(x)+u(x)$, and hence
\[
g(x)+u(x)\in F_I,\ {\rm and}\ g(x)+u(x)\notin U.
\]
Therefore, $F_I$ contains an element outside $U$, it follows that
\[
\dim\, F_I={\rm dim}\,(V_I\cap U)+{\rm dim}\,(F_I/(F_I\cap U) )=(k-p-1)+1=k-p=d_p.
\]
This completes the proof.
\end{proof}

\begin{example}
Let $\mathcal{C}_1$ be a TRS code defined in (\ref{def:C(a,1)}). Let $q=17$, $n=11$, $k=6$, $t=3$, $h=2$, $\eta=5$,  take $\bm{\alpha}=(0,1,2,4,5,7,8,10,12,14,16)$. Let $I=\{3,4,5,8,11\}$, then
\[
P_I(x)=\prod_{i\in I}(x-\alpha_i)=x^5+14x^4+8x^3+14x+9.
\]
Moreover, $p-h=5-2=3$, $k-h-1=6-2-1=3$, we have $c_3=0$.
Therefore, the first condition  is satisfied. And
\[
\eta\sum_{j=0}^{1}c_{k-h+j}q_{t-1-j}=5\left(c_4q_2+c_5q_1\right)=5\left(14\cdot1+9\cdot3\right)=5\cdot41=1.
\]
Therefore, the second condition is also satisfied. At last, since $p=5<k=6$, there is no integer $b$ satisfying $k\le b\le p-1$. Hence the third condition in Proposition \ref{dimFI} is vacuous. With the help of SageMath, we have $\dim\, F_I=2$, it is consistent with Proposition \ref{dimFI}.
\end{example}

We can therefore obtain the complete subcode support weight distribution of $\mathcal{C}_1$, as stated in the following theorem.

\begin{theorem}\label{thm:sswd}
Let $\mathcal{C}_1$ be a TRS code defined in (\ref{def:C(a,1)}). For $1\leq r\leq k$, let
$A_w^{(r)}(\mathcal{C}(\boldsymbol{\alpha},\mathbf{1}))$ denote the $r$-th subcode support weight distribution of $\C_1$. Define
\[
M_p=|\left\{I\subseteq[n]\,:\,|I|=p,\,\dim\, F_I=d_p+1\right\}|,
\]
where $d_p={\rm max}\{k-p,0\}$. Then, for every $0\leq s\leq n$,
\[
A_{n-s}^{(r)}(\C_1)=\sum_{p=s}^{k-1+t}(-1)^{p-s}\binom{p}{s}
\left(\left(\binom{n}{p}-M_p\right)
\begin{bmatrix}
d_p\\ r
\end{bmatrix}_q
+M_p\begin{bmatrix}
d_p+1\\ r
\end{bmatrix}_q
\right),
\]
\end{theorem}

\begin{proof}
By Lemma~\ref{sswd of C}, we have
\[
A_{n-s}^{(r)}(\C_1)=\sum_{p=s}^{n}(-1)^{p-s}\binom{p}{s}
\sum_{\substack{|I|=p}}
\begin{bmatrix}
\dim\, F_I\\ r
\end{bmatrix}_q.
\]
By Proposition~\ref{dimFI}, for every $p$-subset $I$, we have $\dim\, F_I\in\{d_p,d_p+1\}$, and exactly $M_p$ such subsets satisfy $\dim\, F_I=d_p+1$. Hence
\[
\sum_{\substack{|I|=p}}
\begin{bmatrix}
\dim\mathcal F_I\\ r
\end{bmatrix}_q
=\left(\binom{n}{p}-M_p\right)
\begin{bmatrix}
d_p\\ r
\end{bmatrix}_q
+
M_p
\begin{bmatrix}
d_p+1\\ r
\end{bmatrix}_q.
\]
Moreover, every nonzero polynomial in $\mathcal S$ has degree at most
$k-1+t$. Thus $F_I=\{0\}$ whenever $|I|>k-1+t$, and hence
all terms with $p>k-1+t$ vanish. Substituting the above identity into Lemma~\ref{sswd of C} gives the desired result.
\end{proof}

\begin{example}
Let $\mathcal{C}_1$ be a TRS code defined in (\ref{def:C(a,1)}). Let $q=7$, $n=6$, $k=3$, $t=2$, $h=1$, $\eta=3$, and take $\bm{\alpha}=(0,1,2,3,4,6)$.  For a $p$-subset $I\subseteq[n]$, let
\[
F_I=\{f(x)\in S:P_I(x)\mid f(x)\},
\qquad
d_p=\max\{k-p,0\},
\]
and define $M_p=|\{I\subseteq[n]:|I|=p,\ \dim F_I=d_p+1\}|$. With the help of SageMath, we have $M_0=M_1=M_2=M_4=0,\ M_3=2.$
\begin{table}[htbp]
\centering
\caption{The dimensions of $F_I$ and the values of $M_p$}
\label{tab:Mp}
\begin{tabular}{cccc}
\toprule
$p$ & $d_p$ & Distribution of $\dim F_I$ & $M_p$ \\
\midrule
$0$ & $3$ & $\dim F_I=3$ for $1$ subset & $0$ \\
$1$ & $2$ & $\dim F_I=2$ for $6$ subsets & $0$ \\
$2$ & $1$ & $\dim F_I=1$ for $15$ subsets & $0$ \\
$3$ & $0$ &
\begin{tabular}{c}
$\dim F_I=0$ for $18$ subsets,\\
$\dim F_I=1$ for $2$ subsets
\end{tabular}
& $2$ \\
$4$ & $0$ & $\dim F_I=0$ for $15$ subsets & $0$ \\
\bottomrule
\end{tabular}
\end{table}
For $r=1$, Substituting the above values of $M_p$ into Theorem \ref{thm:sswd} gives exactly distribution:
\[
A_3^{(1)}(\C_1)=2,\ A_4^{(1)}(\C_1)=9, A_5^{(1)}(\C_1)=24, A_6^{(1)}(\C_1)=22.
\]
For $r=2$, we have
\[
A_5^{(2)}(\C_1)=6,\ A_6^{(2)}(\C_1)=51.
\]
Finally, for $r=3$, we have
\[
A_6^{(3)}(\C_1)=1.
\]
These results are consistent with the formula in Theorem \ref{thm:sswd}.
\end{example}

\section{Concluding Remarks}

In this paper, we investigated the generalized Hamming weights of twisted Reed--Solomon codes with one and two twists. For TRS codes with one twist, we determined all possible generalized Hamming weights and obtained a complete characterization of the weight hierarchy. In particular, the higher generalized Hamming weights are determined by the minimum distance of the dual code, while explicit criteria were given for the minimum distances of both the code and its dual. These results also recover several known MDS, NMDS \cite{Huang2022},  and $s$-gap AMDS cases \cite{Yang2026}. We further extended the analysis to TRS codes with two twists and determined their complete generalized Hamming weight hierarchy. In this case, the higher generalized Hamming weights are determined by the first two generalized Hamming weights of the dual code, and explicit rank conditions were obtained for the relevant generalized Hamming weights
of both the code and its dual. In addition, for TRS codes with one twist, we derived an explicit formula for the complete subcode support weight distribution.

It would be interesting to extend these results to TRS codes with more
general twist structures and to determine their generalized Hamming
weight hierarchies and subcode support weight distributions. In
particular, obtaining explicit results for codes with an arbitrary
number of twists may provide a more general understanding of the
support structure of twisted Reed--Solomon codes.

\end{document}